\documentclass[12pt]{article}

\usepackage{amsmath, amsfonts, amssymb, amsthm}
\usepackage{dsfont}
\usepackage{enumerate, color}
\usepackage{multido}
\usepackage{url}
\usepackage{mathrsfs}
\usepackage{mathtools}
\usepackage{verbatim}
\usepackage{setspace}
\usepackage{bbm}
 \usepackage{cleveref}
\usepackage{cancel}
\usepackage[normalem]{ulem}

\usepackage{calc}

\usepackage[margin=3cm]{geometry}
\usepackage{tikz}

\newif\iffinal
\finalfalse 
\finaltrue 
\iffinal\else \usepackage[notref,notcite]{showkeys}\fi

\newtheorem{theorem}{Theorem}[section]
\newtheorem{lemma}[theorem]{Lemma}

\newtheorem{assumption}[theorem]{Assumption}

\crefname{theorem}{Theorem}{Theorems}
\crefname{lemma}{Lemma}{Lemmas}
\crefname{proposition}{Proposition}{Propositions}
\crefname{section}{Section}{Sections}
\crefname{assumption}{Assumption}{Assumptions}
\crefname{equation}{Eq.}{Eqs.}

\newcommand{\ba}{\begin{align}}
\newcommand{\ea}{\end{align}}

\newcommand{\E}{\mathbf{E}}

\renewcommand{\Re}{\operatorname{Re}}

\newcommand{\ket}[1]{|{#1}\rangle}
\newcommand{\bra}[1]{\langle{#1}|}

\renewcommand{\d}[1]{\nabla_{#1}}

\renewcommand{\Re}{\mathrm{Re}\, }

\newcommand\otimesal{\mathop{\hbox{\raise 1.6 ex
  \hbox{$\scriptscriptstyle\mathrm{al}$}
\kern -0.92 em \hbox{$\otimes$}}}}
\newcommand\oplusal{\mathop{\hbox{\raise 1.6 ex
  \hbox{$\scriptscriptstyle\mathrm{al}$}
\kern -0.92 em \hbox{$\oplus$}}}}
\newcommand\Gammal{\hbox{\raise 1.7 ex
\hbox{$\scriptscriptstyle\mathrm{al}$}\kern -0.50 em $\Gamma$}}

\newcommand{\caM}{{\mathcal M}}
\newcommand{\caN}{{\mathcal N}}

\newcommand{\caP}{{\mathcal P}}

\newcommand{\caX}{{\mathcal X}}

\newcommand{\bbE}{{\mathbb E}}

\newcommand{\opunit}{\text{1}\kern-0.22em\text{l}}

\renewcommand{\d}{{\mathrm d}}

\newcommand{\beq}{ \begin{equation} }
\newcommand{\beqs}{ \begin{equation*} }
\newcommand{\eeqs}{ \end{equation*} }
\newcommand{\eeq}{ \end{equation} }
\newcommand{\bet}{ \begin{theorem} }
\newcommand{\eet}{ \end{theorem} }

\newcommand{\adjoint}{\mathrm{ad}}
\newcommand{\ad}{\adjoint}

\newcommand{\tr }{\mathrm{tr}}

\newcommand{\n}{\hat{\mathcal{N}}}

\newcommand{\id}{\mathbbm{1}}
\newcommand{\lin}{\mathcal{L}}
\newcommand{\adj}{\mathrm{ad}}

\definecolor{lila}{rgb}{0.6, 0.4, 0.8}

\usepackage{authblk}

\begin{document}

\title{Perturbation Theory for Weak Measurements in Quantum Mechanics, I -- Systems with Finite-Dimensional State Space}
\author[1]{ \small M. Ballesteros}
\author[2]{N. Crawford}
\author[3]{M. Fraas}
\author[4]{J. Fr\"ohlich}
\author[5]{B. Schubnel}
\affil[1]{Department of Mathematical Physics, Applied Mathematics and Systems Research Institute (IIMAS),  National Autonomous University of Mexico (UNAM)}
\affil[2]{Department of Mathematics, Technion}
\affil[3]{Department of Mathematics; Virginia Tech}
\affil[4]{Institut f{\"u}r Theoretische Physik,   ETH Zurich}
\affil[5]{Swiss Federal Railways (SBB)}

\renewcommand\Authands{ and }

\normalsize 
\date \today

\maketitle

\begin{abstract} 
The quantum theory of indirect measurements in physical systems is studied. The example of an indirect measurement of an observable represented by a self-adjoint operator $\mathcal{N}$ with finite spectrum is analysed in detail. The Hamiltonian generating the time evolution of the system in the absence of direct measurements is assumed to be given by the sum of a term commuting with $\mathcal{N}$ and a small perturbation not commuting with $\mathcal{N}$. The system is subject to repeated direct (projective) measurements using a single instrument whose action on the state of the system commutes with $\mathcal{N}$. If the Hamiltonian commutes with the observable $\mathcal{N}$ (i.e., if the perturbation vanishes) the state of the system approaches an eigenstate  of 
$\mathcal{N}$, as the number of direct measurements tends to $\infty$. If the perturbation term in the Hamiltonian does \textit{not} commute with $\mathcal{N}$ the system exhibits ``jumps'' between different eigenstates of 
$\mathcal{N}$. We determine the rate of these jumps to leading order in the strength of the perturbation and show that  if time is re-scaled appropriately a maximum likelihood estimate of $\mathcal{N}$ approaches a Markovian jump process on the spectrum of $\mathcal{N}$, as the strength of the perturbation tends to $0$.
\end{abstract}

\section{Introduction}

Quantum-mechanical models of physical systems interacting with long sequences of probes that are subsequently subject to direct (i.e., projective) measurements are of fundamental interest in studies of quantum filtering and control (see, e.g., \cite{Holevo, Wiseman}) and of the foundations of quantum mechanics (see \cite{vonNeumann, Gisin1984, Diosi1988, Barchielli1991, BaBe, BaBeBe2, FrSchu, BaBeTi, BaBeTi2015, BaBeTi2016}, among others). There is extensive literature on such models, and one may wonder whether something new about these matters can still be added. Put briefly, we propose to study a model of indirect measurements of a weakly \textit{time-dependent} physical quantity that is simple enough that it can be analysed with mathematical precision and, yet, retains realistic features. 

The time evolution of the state of a system interacting with a sequence of probes that are subsequently measured projectively is given by unitary (Schr\"{o}dinger) evolution interrupted by ``state jumps'' triggered by projective measurements of the probes whose out-states are entangled with the state of the system; see Eq. \eqref{jump} below. Recordings of the frequencies of protocols of probe measurements endow the mathematical description of such a system with probabilistic (measure-theoretic) structures. Questions of primary interest to us concern these structures, as well as the resulting stochastic time evolution of states and observables of the system.

Following various technological breakthroughs in the manipulation of small quantum systems \cite{HarocheNobel, WinelandNobel}, experiments approximately described by models of the kind studied in this paper have become possible and are carried out quite routinely, \cite{guerlin, Murch}. An interesting example is a cavity QED experiment \cite{guerlin}: Nearly independent Rydberg atoms, all prepared in the same initial state, are sent, one at a time, through a cavity filled with stationary electromagnetic radiation. The atoms are out of resonance with the modes of the electromagnetic field inside the cavity, so that the probability for emission or absorption of a cavity photon by an atom traveling through the cavity is negligibly small. Yet, during its passage through the cavity the state of the atom is affected by the electromagnetic field inside the cavity, so that it becomes entangled with the state of the field. As a consequence of entanglement a consecutive direct (projective) measurement of an observable associated with the atom then induces a change of the state of the radiation field in the cavity and provides crude information about this state. In what follows, experimental protocols of this type will be referred to as indirect measurements. A detailed analysis of a simple model of a related (solid-state) experiment has been carried out in \cite{BFFS}.

A general theoretical framework for the description of indirect measurements, based on the formalism of ``quantum operations'' developed by Kraus \cite{Kraus}, was introduced by Davies \cite{Davies}. The change of state of the subsystem of interest -- the cavity field in the example discussed above -- induced by the measurement of an observable associated with the probe is encoded in ``jump operators'', $V_\xi$.  Here, the observable being measured has a spectrum denoted by $\caX$ and $\xi \in \mathcal{X}$ is the outcome of the direct measurement of this observable. Corresponding to the probe measurement outcome $\xi$,  the density matrix of the subsystem of interest, $\rho$, undergoes a change described by 
\begin{equation}\label{jump}
\rho \quad \mapsto \quad \frac{V_\xi^* \rho V_\xi}{\tr (V_\xi^* \rho V_\xi)}.
\end{equation} 
When many probes interact with the subsystem of interest, one after another, a protocol, $\underline \xi = \xi_1,\,\xi_2,\dots$, of measurement data recorded at times $t_1,\, t_2, \dots$ results. The measurement times can be deterministic or random. In the model studied in this paper, the probe measurements are made at randomly chosen times and are supposed to take place instantaneously.  Between two consecutive probe measurements the time evolution of the state of the subsystem of interest is unitary and is generated by a Hamiltonian, $\varepsilon H$, where $\varepsilon>0$ is a measure of the strength of the Hamiltonian; i.e.,
\begin{equation}\label{hamiltonian}
\rho \quad \mapsto \quad  e^{-i \varepsilon (t - t_j) H } \rho \,e^{i \varepsilon (t - t_j) H }, \qquad t_j < t < t_{j+1},
\end{equation}
and $\rho$ is the state of the system immediately after the $j^{th}$ probe measurement. We denote by $ \tau_{j+1} = t_{j+1} - t_j$ (for $ j > 1$) and $\tau_1 = t_1$ the interarrival times.
With an initial state $\rho_0$, a measurement protocol $ \underline \xi$, interarrival times $ \underline{\tau} := (\tau_1, \tau_2, \cdots)$ and a positive time $t $, we can thus associate a time-evolved state, $\rho_t(\underline{\tau}, \underline \xi)$, determined by alternatively using \eqref{jump} and \eqref{hamiltonian}. Various discrete and continuous variants of this model have been studied in \cite{Gisin1984, Diosi1988, Adler2001, Barchielli1991, vanHandel2006}. A comprehensive overview is provided in Holevo's book \cite{Holevo}.

If the interarrival times $\tau_{j+1} = t_{j+1}-t_j $ are independent and exponentially distributed with mean $1/\gamma$ then the time evolution of the state obtained by averaging over all possible measurement times and outcomes is a \textit{Lindblad evolution} with generator, $\mathcal{L}_{\varepsilon}$, given by
\begin{equation}
\label{i:1}
\lin_{\varepsilon} \rho = -\frac{i}{\hbar} \varepsilon[H, \rho] + \gamma( \int V_\xi^* \rho V_\xi \d \mu(\xi) - \rho).
\end{equation} 
Here the probability measure $\mu$ is a fixed \textit{a priori} distribution on the configuration space of single-probe measurement outcomes $\xi \in \mathcal{X}$, and the jump operators $V_{\xi}$ are normalized such that
$\int V_{\xi}\,V_{\xi}^*\text{d}\mu(\xi) = {\bf{1}}.$ 
Equation \eqref{i:1} has a natural operational interpretation \cite{AFGG}: With probability $\gamma \d t$, a completely positive operation, $\Phi ( \rho_t) := \int  V_\xi^* \,\rho_t V_\xi \d \mu(\xi)$ is applied on the state, $\rho_t$, of the system during the time interval $[t, t+ \d t)$. The process corresponding to the states $\rho_t(\underline{\tau}, \underline \xi)$ is called an ``unravelling'' of the Lindblad evolution generated by the Lindbladian given in Eq. (\ref{i:1}). The states $\rho_t(\underline {\tau}, \underline \xi)$ appear as integrands in the Dyson-series expansion of 
$\text{exp}(t\mathcal{L}_{\varepsilon}) \rho_{0}$, with $\gamma \int V_{\xi}^{*}\, (\cdot) V_{\xi} \text{d}\mu(\xi)$  viewed as the perturbation.

Motivated by the 2012 Nobel Prize of S. Haroche, a special class of such models, originally introduced in \cite{Maassen}, has recently attracted considerable attention, \cite{BaBe, BabeBe1, BaBeBe2, BFFS, BePel, Amini2011}: `Non-demolition' measurement of a certain observable $\mathcal{N}$ -- in the experiment described in \cite{guerlin} the number of photons trapped in the cavity -- refers to the idealized setting in which \textit{both} the Hamiltonian $H$ and the jump operators $V_\xi$  commute with $\mathcal{N}$. Under repeated non-demolition measurements, it has been observed experimentally, see \cite{guerlin}, and explained theoretically, see \cite{Maassen, vanHandel2004, BaBe, BFFS}, that the state of the system approaches an eigenstate of $\mathcal{N}$.  Moreover, the corresponding eigenvalue of $\mathcal{N}$ is uniquely determined by the measurement protocol $(\underline{ {\tau} }, \underline{\xi})$. 

These results are recalled and generalised to observables with arbitrary (including continuous) spectra in \cite{BCFFS17}. In earlier work \cite{BFFS}, we derived these results using the framework of maximum-likelihood estimates (MLEs).  In this language, the MLEs $\hat{\mathcal{N}}_k$, of $\mathcal{N}$ are constructed from the protocols
$\xi_1,\cdots, \xi_k, k=1,2,3, \dots$ of the  first $k$ measurement outcomes. Thus, from the perspective of \cite{BFFS,BCFFS17}, the aforementioned approach to repeated non-demolition measurements and the phenomenon of purification is a consequence of almost-sure convergence of 
$\hat{\mathcal{N}}_k$, as $k\rightarrow \infty$.

Coming back to the experiment of the Haroche group \cite{guerlin}, the electromagnetic field in the cavity very slowly relaxes to the vacuum state. If $\mathcal{N}$ is identified with the number operator counting photons in the cavity then, apparently, this observable is not strictly time-independent, but shows a slow variation in time. Thus, an indirect measurement of 
$\mathcal{N}$, using a sequence of probes consisting of Rydberg atoms traveling through the cavity is not really a non-demolition experiment. In fact, the value of the photon number estimated on the basis of long sequences of probe measurements will occasionally jump from one value to another one, contrary to the behavior observed in strict non-demolition experiments. In the experiment described in \cite{guerlin} it tends to decrease and approach $0$, as time $t$ tends to infinity.

A natural theoretical strategy for investigating the origin of the behavior seen in this and other related experiments, e.g. \cite{Sayrin}, is to carry out a perturbative analysis of the actual time-evolution around the one corresponding to non-demolition measurements,  with $\varepsilon $ the parameter measuring the strength of the perturbation not commuting with the system obeservable to be measublack indirectly. In this paper we consider Hamiltonian perturbations of non-demolition measurements.  This means that the commutator, 
$\varepsilon [H, \mathcal{N}]$, of the Hamiltonian of the system and the observable $\mathcal{N}$ to be measublack indirectly is taken to be non-zero; see Eq. \eqref{i:1}. In this situation one expects that the density matrix of the system remains close to a rank-one spectral projection onto an eigenstate of $\mathcal{N}$, during most of the time, with occasional jumps from one eigenstate to another one. In the following, we will make this picture precise for a simple model and determine the distribution of such jumps  in terms of physical parameters of the system. Similar results in a continuous-measurement setting have been presented by Bauer, Bernard and Tilloy in a series of papers; see \cite{BaBeTi, BaBeTi2015, BaBeTi2016}. 

Next, we describe the setting of our analysis and explain the main results established in this paper. We will always assume that the state space of the subsystem of interest is finite-dimensional; in the example of cavity QED, this means that only finitely many photons can be confined inside the cavity.  Given a time $t$, let 
$\hat{\mathcal{N}}_t$ be the MLE determined by the probe measurement outcomes obtained in the time interval $(t, t+T)$, for some $T>0$ to be chosen. Recall that 
$\varepsilon$ denotes the strength of the Hamiltonian perturbation, see \eqref{i:1}. We show that,  for $T =  \alpha | \log  \varepsilon |$- for sufficiently large $\alpha$, and after re-scaling time by 
$\varepsilon^{-2}$, the process  $\hat{\mathcal{N}}_{\varepsilon^{-2}t}$ converges in law to a Markov jump process on the spectrum of $\mathcal{N}$, in the natural Skorokhod topology, as $\varepsilon $ tends to zero. 

In the limiting Markov process, the rates of transitions between different eigenstates of $\mathcal{N}$ are given as follows: Let $\mathcal{N} = \sum_{\nu\in \textrm{Spec}(\caN)} \nu P_\nu$ be the spectral decomposition of  the system-observable and assume that all its eigenvalues $\nu$ are non-degenerate.  Then the matrix elements of the generator, $Q$, of the limiting process are given by the formula 
\begin{equation}
\label{i:2}
Q(\nu',\nu)=\frac{2}{\gamma \hbar^2} \Re  \left( - \frac{|\bra{\nu'} H \ket{\nu}|^2}{ \int \d \mu(\xi) V_\xi(\nu') \overline{V}_\xi(\nu) - 1} \right), \quad \text{  for   }\,\nu' \not= \nu,
\end{equation}
with $\bra{\nu} V_\xi \ket{\nu'} = \delta_{\nu \nu'} V_{\xi}(\nu)$. Among our results concerning the convergence of the quantum evolution towards a Markov jump process is the following theorem: 
\begin{equation*}
\lim_{\varepsilon \rightarrow 0} \,\langle \nu \vert e^{\varepsilon^{-2}t \mathcal{L}_{\varepsilon}} \rho_{0} \vert \nu \rangle = \big(e^{tQ} \pi_{\rho_0}\big)(\nu),
\end{equation*}
where $\pi_{\rho_0}(\nu):= \langle \nu\vert \rho_{0} \vert \nu \rangle$, with $\vert \nu \rangle$ the eigenstate of 
$\mathcal{N}$ corresponding to the eigenvalue $\nu$; (see Theorem \ref{thm:1}).
 
 There is an interesting technical caveat to be noted here:  The strongest sense imaginable in which a limit law may hold is that the stochastic process defined by posterior density matrices converges in Skorohod space to a Markov process on rank-one projections onto the eigenstates of $\mathcal{N}$, with transistion rates given by \Cref{i:2}.  It has been argued in \cite{BaBeTi} that this strong convergence cannot hold in general.  Our result circumvents this (technical) problem in that it is claimed that if the process of posterior density matrices is averaged over mesoscopic time windows convergence does in fact hold.  In a forthcoming paper, we will characterise circumstances under which the ``spiky'' behaviour observed in \cite{BaBeTi} occurs. 
 
 In future work, we plan to generalize our results by replacing an observable $\mathcal{N}$ with a finite point spectrum by a $d$-tuple, $\vec{\mathcal{Q}}$, of commuting system-observables with continuous spectra, for example
 $\sigma(\vec{\mathcal{Q}}) = \mathbb{R}^{d}$, (see \cite{BCFFS17}). Assuming that the Hamiltonian $H$ of the system does not commute with the operators $\vec{\mathcal{Q}}$, we may expect that the quantum-mechanical evolution of the state of the system approaches one corresponding to a stochastic process on $\mathbb{R}^{d}$ with non-vanishing drift given by a vector field on $\mathbb{R}^{d}$. Insights of this type are of interest in connection with attempts to render Mott's analysis of particle tracks mathematically respectable. Results  relevant for our purposes have been proven in \cite{BaDuKo, BaDu}.

\subsection{Summary of contents}
In the next section we describe the setting which this paper is based on in some detail, and we summarize our main results. In particular, we construct the measure space of measurement protocols $(\underline{\tau}, \underline \xi)$ of direct probe measurements and equip it with a measure pblackicting the frequencies of such protocols. We introduce the maximal likelihood estimate, $\hat{\mathcal{N}}$, to be used to prove convergence of $\hat{\mathcal{N}}$ to a jump process with the rate given in Eq.~(\ref{i:2}). 

In Sect. \ref{sec:prelim}, we present the proofs of our results. Some auxiliary estimates are deferblack to Appendix \ref{App:A}. In Appendix \ref{App:D} we list all relevant notation.

We will work in units in which $\hbar =1$ and the rate $\gamma$ of the Poisson process of interarrival times is unity. 
\\

{\bf{Acknowledgements.}} J. Fr\"ohlich thanks M. Bauer, D. Bernard and A. Tilloy for useful information about their results.  M. Ballesteros is a fellow of the Sistema Nacional de Investigadores (SNI). His research is partially supported by the projects PAPIIT-DGAPA UNAM IN102215 and SEP-CONACYT 254062.  

\section{Main result}
\subsection{Setup and notation}
\subsubsection{Hilbert space and operators}
Let $\mathcal{H}$ be a finite-dimensional Hilbert space. We denote the algebra of bounded linear maps from $\mathcal{H}$ to itself by $\mathcal{B}(\mathcal{H})$.  The usual operator norm on $\mathcal{B}(\mathcal{H})$ is denoted by $\| .\|$, and the Hilbert-Schmidt norm is denoted by $\|. \|_2$. For ``super-operators'' $ \mathcal{O}:\mathcal{B}(\mathcal{H}) \rightarrow \mathcal{B}(\mathcal{H})$, we use the super-operator norm $\|. \|_{2, \text{op}}$ defined by 
\begin{equation} \label{nomame}
\| \mathcal{O} \|_{2, \text{op}} := \underset{\| X \|_{2} =1 }{\sup} \| \mathcal{O}( X )\|_{2}. 
\end{equation}
 A natural example of a super operator which figures prominently below is  $\adj_H(X):=HX-XH$, for $X \in \mathcal{B}(\mathcal{H})$.  Here,  $H$ is an arbitrarily chosen self-adjoint operator acting on $\mathcal{H}$.

\subsubsection{Probe measurements}
Given a measure space $(\mathcal{X}, \sigma)$ and a probability measure $\mu$, we consider a measurable family of bounded complex-valued functions $V_\xi(\nu) : \nu \in \sigma(\mathcal{N})  \to \mathbb{C}$ satisfying the normalisation condition
\begin{equation}
\label{normalis}
\int_{\mathcal{X}} |V_\xi(\nu)|^2 \d \mu(\xi) = 1, \quad \mbox{for all} \quad \nu \in \sigma(\mathcal{N}).
\end{equation}
The set $\mathcal{X}$ represents all possible outcomes of direct measurements taken of a probe that has previously interacted with the system of interest.  The time evolution of the probe during its interaction with the system is affected by the value of an observable, henceforth denoted by $\mathcal{N}$, represented by a self-adjoint operator that we also denote by $\mathcal{N}$.  
Following \cite{Maassen}, we  define a family of operators $V_{\xi}, \xi \in \mathcal{X},$ acting on the Hilbert space 
$\mathcal{H}$ of the system, with $V_\xi $ describing the effect of a probe measurement with outcome $\xi$ on a state of the system. The operators $V_{\xi} \equiv V_{\xi}(\mathcal{N})$ depend on the observable $\mathcal{N}$. Thus, performing direct measurements on a sequence of such probes may yield information on the value of 
$\mathcal{N}$. It will always be assumed that the interaction of a probe with the system does not affect the value of $\mathcal{N}$.  

Let $P_\nu$ denote the spectral projection of $\mathcal{N}$ corresponding to the eigenvalue $\nu$.  The operators $V_{\xi}$ have the form
\begin{equation}
\label{eq:spectral_decomposition_of_C}
V_\xi = \sum_{\nu \in \sigma(\mathcal{N})} V_\xi(\nu) P_\nu.
\end{equation}
We introduce the random super-operators 
$$
\Phi_\xi :   B(\mathcal{H}) \to   B(\mathcal{H}),  \,\,  \,\,\Phi_\xi(X) :=  V_\xi^* X V_{\xi} . 
$$ 
The map $\Phi_\xi$ encodes the effect of a probe measurement with outcome $\xi \in \mathcal{X}$ on the state of the system. 

\subsubsection{Time evolution} \label{ET}

To describe the effects of repeated measurements on the state of the system, we introduce the space $\Xi \equiv [0, \infty)^\mathbb{N} \times \mathcal{X}^\mathbb{N}$ of infinite sequences \mbox{$(\underline\tau, \underline\xi)$} of outcomes of direct probe measurements $\underline{\xi} \equiv \xi_1, \xi_2, \dots$ separated by times $\underline{\tau}\equiv \tau_1,\tau_2,\dots$. Here $\tau_j$ is the time between the $ (j-1) $st measurement and the
$j$th measurement, for $ j >1$; if $j=1$ it is the time when the first measurement happens.  We equip this space with the standard sigma algebra, $\mathcal{F}$, generated by cylinder sets. Let $\mathbb{P}$ be a probability measure on $(\Xi, \mathcal{F})$ for which the coordinate functions $\{\tau_j, \xi_k\}_{j,k=1}^\infty$ are independent, the times $\tau_j$ are exponential random variables with mean $1$, and the measurement outcomes $\xi_k$ are distributed according to the measure $\mu$. 
 We denote by $\mathbb{E}$ the expectation value associated to $\mathbb{P}$.  In Appendix \ref{App:0} we give precise definitions and present additional properties of the measure space that  we have just defined. 

The random variables $(t_j)_{j \in \mathbb{N}}$, defined by 
$$
t_j(\underline{\tau}) \equiv t_j := \sum_{i = 1}^j \tau_i,
$$    
represent the measurement times: at time $ t_j $ the $j$th measurement takes place.  The process of counting the number of measurements up to time $s$,
$$
N_s(\underline{\tau}) \equiv N_s : = \sup \big \{ n :  t_n(\underline{\tau}) \leq s   \big \}, 
$$ 
is a rate-one Poisson process on $[0, \infty)$.\\
 Let $0 < \varepsilon \ll 1$ and $\varepsilon H$ be the Hamiltonian of the system, assumed to be a Hermitian operator. 
The time evolution of the system is represented by the operator valued random variable (see also \eqref{misto})
\begin{equation}
\label{prop}
{\sigma}_\varepsilon^{(s,u)}(\underline{\tau}, \underline \xi) :=  e^{-i \varepsilon (s-t_{N_s}) \adj_H}  \Phi_{\xi_{N_s}}   \dots   e^{-i \varepsilon (t_{N_{u}  +2} - t_{N_{u}+1}) \adj_H} \Phi_{\xi_{N_{u}+1}} e^{-i \varepsilon (t_{N_{u}+1}-u) \adj_H}. 
\end{equation}  
Up to a normalization, the super-operator ${\sigma}_{\varepsilon}^{(s,u)}\equiv {\sigma}_{\varepsilon}^{(s,u)}
(\underline{\tau}, \underline{\xi})$ maps the state of the system at time $u$ to its state at time $s$. The following property is a consequence of the definition above 
\begin{equation}
\label{pgp}
 {\sigma}_\varepsilon^{(s, u)} {\sigma}_\varepsilon^{(u,v )} = {\sigma}_\varepsilon^{(s, v)},
\end{equation}
for times  $0\leq v< u < s$.

 For a measurable set $E$, we set
\begin{equation}\label{para}
 {\sigma}_\varepsilon^{(s, u)}(E) : = \int_E  {\sigma}_\varepsilon^{(s, u)} \d \mathbb{P}.    
\end{equation}
 In Appendix \ref{App:0} we show some important properties of the super-operators in \eqref{para}, in particular the factorization property in Eqs. \eqref{Factor} and \eqref{factorization}  manifesting Markovianity of the process.  Informally, Eq.~\eqref{factorization} states that for sets $E_1$, $E_2$ depending on measurement results in interval $(0,u]$ resp. $(u,s]$ it holds,
 \begin{equation}
 \label{factorization}
 {\sigma}_\varepsilon^{(0, s)}(E_1 \cap E_2) = {\sigma}_\varepsilon^{(u, s)}(E_1) {\sigma}_\varepsilon^{(0, u)}(E_2).
 \end{equation} 
 
Every state (density matrix) $\rho$ of the system gives rise to a probability measure $\mathbb{P}^\varepsilon_\rho$ on $(\Xi, \mathcal{F})$ defined by
\begin{equation}
\label{eq:2p}
\mathbb{P}^\varepsilon_\rho(E) := \tr({\sigma}_\varepsilon^{(s,0)}(E)[\rho] ).
\end{equation}
We denote by $\mathbb{E}_\rho^\varepsilon[\cdot]$ the associated expectation. Given a point $(\underline{\tau}, \underline{\xi}) \in \Xi$, the posterior state at time $s$ is defined by
\beq
\label{E:rhos}
\rho_s(\underline{\tau},\, \underline \xi):= \frac{{\sigma}_\varepsilon^{(s,0)}(\underline{\tau},\, \underline \xi)[ \rho] }{\tr({\sigma}_\varepsilon^{(s,0)}(\underline{\tau},\, \underline \xi)[ \rho])}.
\eeq
This posterior state represents the state of the system, given the measurement outcomes  $\underline \xi$ at times $ t_1, t_2, \cdots $.

\subsubsection{Log-likelihood function and maximum likelihood estimator}
\label{loglike}
In the non-demolition situation, i.e., for $\varepsilon = 0$, the time-evolution of the system is trivial between any two consecutive probe  measurements. The function $f(\cdot \vert \nu)$, defined by
\begin{equation} \label{pto}
f(\xi | \nu) := |V_\xi(\nu)|^2,
\end{equation} 
then has the meaning of a conditional probability distribution, and the equation
\begin{equation}
\mathbb P^{0}_\rho(F_s) = \frac{e^{-s} s^{k}}{k!} \sum_{\nu \in \sigma(\mathcal{N})}  \mu_\nu(\Delta_1) \dots \mu_\nu(\Delta_k) \tr(P_\nu \rho),
 \label{eq:5}
\end{equation}
with 
\beq
\tag{14'}\label{E:munu}
\mu_\nu(\Delta_j) := \int_{\Delta_j}  f(\xi|\nu) \d \mu(\xi)
\eeq
holds for any product set 
$F_s = \{ N_s =k \} \times \Delta_1\times \Delta_2 \dots \times \Delta_k $; (we recall that the operators $P_{\nu}$ are the spectral projections of the observable $\mathcal{N}$).
Equation~(\ref{eq:5}) may be interpeted as the \mbox{de Finetti} decomposition \cite{Finetti} of the measure $\mathbb P^{0}_\rho$.\\

The theory of indirect measurements is  closely linked to parameter estimation in statistics \cite{BCFFS17}. In particular, if 
$\varepsilon = 0$ it is natural to regard $\nu$ as an unknown quantity to be estimated based on measured data, namely the measurement outcomes $\xi_i$.  We will therefore introduce a consistent estimator.  The  log-likelihood, or maximum likelihood, estimator is a natural choice, it is well known that if the measures $\mu_\nu$  are ``identifiable'', i.e., if $\mu_{\nu} \neq \mu_{\nu'}$, for $\nu \neq \nu'$, then it converges to the true value of the parameter $\nu$, as the number of data points tends to infinity. 

One idea used in this paper  is that this estimator can also be used for small but non-zero $\varepsilon$.   Choosing a ``sampling time'' $T>0$, we  introduce the log-likelihood function
\begin{equation}
\label{eq:likelihood}
l_{s}^T(\nu | \underline{\xi} )  := \frac{1}{   N_{s + T} - N_s  }  
 \sum_{j=N_s+1}^{N_{s + T}} \log f(\xi_{j}| \nu ).
\end{equation}
Further, let
\begin{equation}
\label{estin}
\n_s(T) := \underset{\nu \in \sigma(\mathcal{N})}{\mathrm{argmax}}  \text{ }l_{s}^{T}(\nu | \underline{\xi} ).
\end{equation}

Note that, for a given sequence $\underline \xi$ of probe measurement outcomes, there may be more than one point $\nu$ in the spectrum of $\mathcal{N}$ for which the right side of Eq. \eqref{estin} is maximized. If such an ambiguity arises we define $\n_s(T)$ according to some agreed-upon rule.  However, with respect to $\mathbb{P}_\rho^0$, the probability of an ambiguous sample  tends to $0$ exponentially fast in $T$. 

In order to avoid to have to cope with short time fluctuations, which appear to be inherent in the process studied here,
it turns out to be convenient to consider times $s$ that are multiples of the sampling time T, and we therefore introduce the process
\begin{equation}\label{proc}
\mathcal{M}_{j T} = \n_{j T}(T), \quad \mbox{for} \quad j \in \mathbb{N},
\end{equation}
and extend the definition of $\mathcal{M}_t$ to all $t \geq 0$ by declaring it to be  constant on the intervals $[j T, (j+1) T)$. 

\subsection{Statement of the result}
In this section we state the assumptions upon which our analysis rests and then describe our main results. Some key ideas of the proofs are sketched in the next subsection.

\begin{assumption} We require the following hypotheses:
\label{ass:LLN}
\begin{enumerate}
\item  The spectrum of the observable $\mathcal{N}$ is non-degenerate.
\item  The measures $\mu_{\nu}$ introduced in \eqref{eq:5} are identifiable; (i.e., $\mu_{\nu} \neq \mu_{\nu'}$ if $\nu \neq \nu'$).
\end{enumerate}
\end{assumption}

\noindent We introduce a continuous-time Markov (jump) process, $Y_s$, on the spectrum of 
$\mathcal{N}$ by specifying its transition function $\Gamma_t := e^{tQ}$ and its initial probability distribution $\pi_{\rho}(\nu)$ at time $s=0$. The latter is given by $\pi_{\rho}(\nu):=\langle \nu | \rho | \nu \rangle$, and the transition function of the process has a generator given by the (transition-rate) matrix
\begin{equation}
\label{Qdef}
Q(\nu',\nu) := \left\{ \begin{array}{lcr}
									 2 \Re\left( \sum_{\beta \neq \nu} \frac{|\bra{\beta} H \ket{\nu}|^2}{\int_{\mathcal{X}}  \d \mu(\xi) V_\xi(\beta) \overline{V}_\xi(\nu) -1}  \right) & \mbox{for} & \nu = \nu' \\[3mm]
									2 \Re\left(  -\frac{|\bra{\nu'} H \ket{\nu}|^2}{\int_{{\mathcal{X}}  } \d \mu(\xi) V_\xi(\nu') {\overline{V}}_\xi(\nu) -1}  \right) & \mbox{for} & \nu \neq \nu'.								\end{array} \right.
\end{equation}
By definition of $Y_s$, we have that $\text{Pr}(Y_{s+h}=\nu | Y_{s}=\nu')=\delta_{\nu \nu'} + Q(\nu',\nu) h + o(h)$.
\\

We are now prepared to state our main results in the form of two theorems. Assumption \ref{ass:LLN} will always be required.

\vspace{4mm}

\begin{theorem}
\label{thm:1}
 For an arbitrary initial state $\rho$ of the system, we have that
\begin{equation*}
\lim_{\varepsilon \rightarrow 0} \,\langle \nu \vert e^{\varepsilon^{-2} s \mathcal{L}_{\varepsilon}} \rho \vert \nu \rangle = \big(\Gamma_{s} \pi_{\rho}\big)(\nu).
\end{equation*}
\end{theorem}

\vspace{4mm}

With the help of additional assumptions we are able to describe the process in a more detailed manner. We recall that $P_\nu$ is the spectral projection of $\mathcal{N}$ associated to the eigenvalue $\nu$ (see Eq.~(\ref{eq:spectral_decomposition_of_C})), and that $\rho_s$ denotes the posterior state at time $s$ introduced in Eq. \eqref{E:rhos}.

\begin{theorem}
\label{thm:case1}
Given an arbitrary $\varepsilon>0$, we choose a sampling time $T=T(\varepsilon) = \alpha| \log(\varepsilon)| $, with 
$\alpha>0$. There exists a constant $\alpha_0>0$ such that, for $\alpha \geq \alpha_0$, the following claims hold true.
\begin{enumerate}[(a)]
\item
\label{TP}
With respect to the measures $\mathbb P^{\varepsilon}_\rho$ defined in \eqref{eq:2p}, the process  
$\mathcal{M}_{\varepsilon^{-2} s}$ on the spectrum of $\mathcal{N}$ introduced in \eqref{proc} converges in law to the continuous-time Markov chain $Y_s$, as $\varepsilon\rightarrow 0$.
\item \label{TDM}
There exists a constant $C$ - that might depend on $\alpha$ -  such that, for sufficiently small $\varepsilon$ (depending on the choice of $\alpha$) and $s > 2 \varepsilon^2 T(\varepsilon)$,  
\begin{equation}\label{ineq_state}
  \bbE^{\varepsilon }_{\rho}\Big [  \|\rho_{\varepsilon^{-2} s}-P_{\hat{\mathcal{N}}_{\varepsilon^{-2} s}} \|_{2}   \Big ]  \leq C  \varepsilon |\log \varepsilon |^{1/2 }.
\end{equation}  
\end{enumerate}     
\end{theorem}

An explicit estimate of the value of $\alpha_0$ in terms of data related to the instrument $\Phi_\xi$ will be given in the proof of the theorem.

\subsection{Plan of the proof}
The proofs of Theorems \ref{thm:1} and  \ref{thm:case1} are divided into three parts, which we sketch in this section. 

In the first part, we show that the dynamics of the state $\rho_s$ arises from an unravelling of the dynamical semigroup generated by the Lindbladian 
\begin{equation}
\label{eq:Lin}
\lin_\varepsilon \rho = -i \varepsilon [H,\rho] + \Phi(\rho) - \rho, \quad  \text{where } \Phi(\rho) := \int_{\mathcal{X}}  V_\xi^* \rho V_\xi \d \mu(\xi),
\end{equation}
see Subsect. \ref{unra}. The time evolution generated by $\lin_\varepsilon $ leads to decoherence  over the spectrum of $\mathcal{N}$. More precisely, using a series of inequalities, we prove that  the Lindbladian time evolution maps any initial state into the subspace $\mathcal{P} B(\mathcal{H})$, as time tends to $\infty$, where the projection $\mathcal{P}$ is defined by
\begin{equation}\label{pndjo}
\mathcal{P} X:= \sum_{\nu \in \sigma(\mathcal{N})} P_{\nu} X P_{\nu}, \quad \mbox{for }  X \in B(\mathcal{H}),
\end{equation}
and $P_{\nu}:= \vert \nu \rangle \langle \nu \vert$ are the (rank-one) spectral projections of $\mathcal{N}$.
On this subspace and for time scales of order $\mathcal{O}(\varepsilon^{-2})$, the Lindbladian evolution is then shown to be norm-close to the time evolution described by the transition function $(\Gamma_t)_{t\geq 0}$ with generator given by the transition-rate matrix $Q$ introduced in Eq. \eqref{Qdef}; see Theorem~\ref{thm:1}.  Precise statements can be found in Lemma \ref{lem:5.1} of  Subsection  \ref{dyave}.

The second part of the proof uses  the fact, based on Assumption  \ref{ass:LLN}, that, for non-demolition measurements, i.e., when $\varepsilon=0$, the state of the system ``purifies'' (i.e., converges) exponentially fast to an eigenstate of $\mathcal{N}$, as the length of the probe-measurement protocol tends to infinity. This ``purification'' is expected to still arise on intermediate time scales for the perturbed dynamics, as long as $\varepsilon$ is very small.  To make this precise, we have to judiciously choose a sampling time $T(\varepsilon)$: it should not be too large, because the term $\varepsilon \text{ad}_H$ plays against large-time purification. Yet, $T(\varepsilon)$ should not be chosen too small either, because one wants to make sure that our estimator $\hat{\mathcal{N}}_{.}(T)$ represents the spectral parameter $\nu$ of the observable $\mathcal{N}$ accurately, which will only be the case if, in a time interval of length $T(\varepsilon)$, many probe measurements will typically be made. It will turn out that the choice  $T = \alpha \log \varepsilon$, for an appropriately chosen constant $\alpha>0$, is adequate.

In the third part of the proof of Theorem \ref{thm:case1} the results of the previous two parts are put together.  A standard criterion for  convergence in Skorokhod space needed to establish Item ($a$) of Theorem \ref{thm:case1} is that, in order to prove convergence in law of a collection of processes, one must first show finite-dimensional convergence and then prove tightness.  Finite-dimensional convergence for the law of $(\caM_{{\varepsilon}^{-2}s_1}, \dotsc, \caM_{{\varepsilon}^{-2}s_n})$, with $0<s_1<....<s_n \leq1$, is a relatively straightforward extension of the convergence result for a single time. It is formulated precisely in Lemma \ref{lem:5.7}. 

On Skorohod space, it is well known that establishing tightness of a sequence of probability measures is equivalent to showing that, in a uniform sense, these measures do not put mass on paths that jump very often.  To prove this property in our context requires an intricate, but elementary, analysis of the propagator ${\sigma}^{(s, 0)}_{\varepsilon}$ for small times $s$. 

 Inequality \eqref{ineq_state}, i.e., Item ($b$) of Theorem \ref{thm:case1}, takes relatively little effort; it follows directly from a comparison of the non-demolition dynamics with the perturbed dynamics; see Subsect. \ref{proofb}.

\section{Proofs of Theorems \ref{thm:1} and \ref{thm:case1}}
\label{sec:prelim}
\subsection{Properties of averaged dynamics}
\label{avera}
In this subsection we establish some spectral estimates concerning the averaged dynamics that will be used in our proofs of Theorems \ref{thm:1} and \ref{thm:case1}.
\subsubsection{Unravelling of the Master Equation}
\label{unra}
We first show that the stochastic dynamics, 
$\big({\sigma}_\varepsilon^{(t,s)}(\underline{\tau}, \underline \xi)\big)_{0<s<t}$, defined in Eq.~(\ref{prop}) is an unravelling of the dynamics generated by the Lindbladian $\lin_\varepsilon$ of Eq.~(\ref{eq:Lin}). We recall that ${\sigma}^{(t,s)}_\varepsilon(\Xi)  = \mathbb{E}[{\sigma}_\varepsilon^{(t,s)}]$.
\begin{lemma}
\label{prop:unraveling}
For $t \geq s$,
$$
\mathbb{E}[{\sigma}_\varepsilon^{(t,s)}] = \exp((t-s)\lin_\varepsilon),
$$
where $\lin_\varepsilon$ is the Lindblad operator defined  in Eq. \eqref{eq:Lin}. 
\end{lemma}

\begin{proof}
In Appendix \ref{App:0} we have claimed that ${\sigma}_\varepsilon^{(t,s)}$ and ${\sigma}_\varepsilon^{(t-s,0)}$ have the same distribution. It is therefore enough to prove the claim for $s=0$. The expectation of ${\sigma}_{\varepsilon}^{(t, 0)}$ can be expressed as
\begin{align*}
 \mathbb{E} [{\sigma}_{\varepsilon}^{(t, 0)}] = \sum_{k  }  \mathbb{E} [ \mathds{1}_{ \{ N_t = k \}}{\sigma}_{\varepsilon}^{(t, 0)} ]   
&   =  \sum_{k=0}^\infty     \underset{   t_{k} \leq  t }{\int} e^{- t}  e^{-i(t -t_{ k}) \varepsilon \adj_H} \Phi e^{-i \tau_{ k} \varepsilon \adj_H}  
\cdots  
\\ & \hspace{3cm}  \cdots   \Phi e^{-i\tau_{ 1 }  \varepsilon \adj_H} 
   d  \tau_1  \cdots d \tau_{ k } .
\end{align*}

Moreover, differentiating $ \mathbb{E} ({\sigma}^{(t,0)}_\varepsilon)$ with respect to $t$ yields
\begin{align*}
\partial_t \mathbb{E}[{\sigma}^{(t,0)}_\varepsilon] &= (-i \varepsilon \adj_H - 1)\mathbb{E}[{\sigma}^{(t,0)}_\varepsilon] \\ 
&\quad + \sum_{k=0}^\infty \int_{t_{k-1} \leq t} \Phi e^{-i(t-t_{k-1}) \varepsilon \adj_H} \dots \Phi e^{-i\tau_1 \varepsilon \adj_H}   d \tau_1 \dots d  \tau_{k-1} 
\\ &= (-i \varepsilon \adj_H - 1 + \Phi)\mathbb{E}[{\sigma}^{(t,0)}_\varepsilon] = \lin_\varepsilon \mathbb{E}[{\sigma}^{(t,0)}_\varepsilon].
\end{align*}
The operator $\lin_\varepsilon$ is manifestly a Lindbladian.
\end{proof}

\subsubsection{Properties of the averaged evolution} 
\label{dyave}
Recall that we have defined the projection
\begin{equation}\label{pen}
\mathcal{P}X = \sum_{\nu \in \sigma(\mathcal{N})} \mathcal{P}_\nu X,   \quad \text{  where   }\,\,\,\,
 \mathcal{P}_\nu X :=  P_{\nu} X P_{\nu} .
\end{equation} 
 The image of $\mathcal{P}$ is the subspace of matrices corresponding to the eigenvalue $0$ of the Lindbladian $\mathcal{L}_0$.  We set $\mathcal{P}_{\perp}:= 1-\mathcal{P}$. The index $\perp$ refers to the  scalar product 
\begin{align}\label{IN}
\langle A,B \rangle:= \tr(A^* B)
\end{align} 
 with respect to which $\lin_0$ is a normal operator.  
 For $\varepsilon=0$, we have that
\begin{equation}
\label{eq:decoupling}
e^{t \lin_0} \mathcal{P} =\mathcal P, \quad \| e^{t \lin_0} \mathcal{P}_\perp \|_{2, \text{op}}  \leq e^{-g_{sp}t},
\end{equation}
 where the spectral ``gap'' $g_{sp}$ is defined by
$$
g_{sp} = \min_{\nu \neq \nu'} \Re \left( 1 - \int_{\mathcal{X}} \d \mu(\xi) \overline{V_\xi}(\nu) V_\xi(\nu') \right)>0,
$$
the expression within parentheses is the negative of the eigenvalue of $\lin_0$ corresponding to the eigenvector $\ket{\nu} \bra{\nu'}$. By Assumption~\ref{ass:LLN} we have that, for $\nu \neq \nu'$,
$$
\int_{\mathcal{X}}  |\overline{V_\xi}(\nu) V_\xi(\nu')| \d \mu(\xi) <1,
$$
and, since the state space of the system is finite-dimensional, we have that
\begin{equation}
\label{Eq:g_def}
g:=1- \max_{\nu \neq \nu'} \int_{
\mathcal{X}} |\overline{V_\xi}(\nu) V_\xi(\nu')| \d \mu(\xi) >0.
\end{equation}
Clearly $g_{sp} \geq g$. In order to avoid introducing too many constants, we use the letter $g$ in all spectral estimates below. In the next lemma, estimates on the averaged evolution for times of order $\varepsilon^{-2}$, for small $\varepsilon$, are presented. 

\begin{lemma}
\label{lem:5.1}
For any $t \geq 0$, we have that 
\begin{equation}
\label{ineq11}
\|e^{t \lin_\varepsilon} \mathcal{P}_\perp \|_{2, \text{op}} \leq e^{-tg} + \frac{2 \varepsilon \|H \|}{g}, \quad \|\mathcal{P}_\perp e^{t \lin_\varepsilon} \|_{2, \text{op}} \leq e^{-tg} + \frac{2 \varepsilon \|H \|}{g},
\end{equation}
and, for $g > 4 \varepsilon \|H \|$, 
\begin{equation}
\label{ineq12}
\big \| \mathcal{P} \left[ \exp(\varepsilon^{-2} s \lin_\varepsilon)- \exp(s Q) \right] \mathcal{P} \big \|_{2, \text{op}} \leq 16 \varepsilon^2 \frac{ \|H\|^2}{g^2} + 48 s \varepsilon \frac{\|H\|^3}{g^2}.
\end{equation}
\end{lemma}

\begin{proof}
Notice that $ \mathcal{L}_{\varepsilon} $ is a dissipative operator, i.e., 
\begin{align}\label{dis}
  2 \Re   \langle \mathcal{L}_{\varepsilon} X,  X \rangle =   2 \Re   \langle  X, \mathcal{L}_{\varepsilon} X \rangle = \langle \mathcal{L}_{\varepsilon} X, X \rangle +  \langle  X, \mathcal{L}_{\varepsilon} X \rangle \leq 0,
\end{align}  
for every matrix $X$, with $\langle \cdot, \cdot \rangle$ as in \eqref{IN}; (this is shown by direct calculation). That  $ \mathcal{L}_{\varepsilon} $ is dissipative is equivalent to $ e^{ t  \mathcal{L}_{\varepsilon} } $ being a  Hilbert-Schmidt contraction semigroup; (see Theorem 1.1, Chapter 7, in \cite{Davies}).

We now prove the first inequality in \eqref{ineq11}.  For any $t \geq 0$, we have that
\begin{equation}
\label{eas}
e^{t \lin_\varepsilon} = e^{t \lin_0} + \int_0^t  \partial_u(e^{(t-u) \lin_0} e^{u \lin_\varepsilon}) du =   e^{t \lin_0}  + \int_0^t e^{(t-u) \lin_0} (\lin_\varepsilon - \lin_0) e^{u \lin_\varepsilon} \d u.
\end{equation}
Multiplying by $\mathcal{P}_\perp$ from the left and taking the norm we obtain that 
\begin{align*}
\|\mathcal{P}_\perp e^{t \lin_\varepsilon}\|_{2, \text{op}} &\leq  \|\mathcal{P}_\perp e^{t \lin_0}\|_{2, \text{op}} + \int_0^t \|\mathcal{P}_\perp e^{(t-u) \lin_0} (\lin_\varepsilon - \lin_0)\|_{2, \text{op}}  \d u. \\
&\leq e^{-gt} + 2 \varepsilon \|H\| \int_0^t e^{-(t-u) g} \d u \leq e^{-tg} + \frac{2 \varepsilon \|H\|}{g}.
\end{align*}
Note that the first inequality follows from the Eq. \eqref{eas} -Duhamel formula, using that $\lin_\varepsilon$ is the generator of a Hilbert-Schmidt contraction semigroup, and we use that  $\|\adj_H\|_{2,op} \leq 2 \|H \|$ to prove the second inequality. 

We omit the proof of the second inequality in \eqref{ineq11} which is analogous.

\noindent The proof of inequality \eqref{ineq12} is more involved.  Sandwiching Eq. \eqref{eas}  by $\mathcal{P}$ and using that $\lin_0 \mathcal{P} = 0$, we find 
\begin{align*}
\mathcal{P} e^{t \lin_\varepsilon} \mathcal{P}&= \mathcal{P} + \int_0^t \mathcal{P} \lin_\varepsilon e^{u \lin_\varepsilon} \mathcal{P} \d u  = \mathcal{P} + \int_0^t \mathcal{P} \lin_\varepsilon \mathcal{P} e^{u \lin_\varepsilon} \mathcal{P}  \d u + \int_0^t \mathcal{P} \lin_\varepsilon \mathcal{P}_\perp e^{u \lin_\varepsilon} \mathcal{P} \d u.
\end{align*}
The normalization condition \eqref{normalis} on the operators $V_{\xi}(\mathcal{N})$ implies that  
$\mathcal{P}  \Phi \mathcal{P} \rho= \mathcal{P} \rho$, for all $\rho$, and hence that  $\mathcal{P} \lin_\varepsilon \mathcal{P} =0$. Introducing the notations
$$W_t := \mathcal{P} e^{t \lin_\varepsilon} \mathcal{P}, \qquad \lin_\varepsilon^\perp:= \mathcal{P}_\perp \lin_\varepsilon \mathcal{P}_\perp,$$
we may rewrite the last equation as 
\begin{equation}
\label{tointegrate}
W_t= \mathcal{P}  + \int_0^t \int_0^u \mathcal{P} \lin_\varepsilon    \mathcal{P}_\perp e^{(u-v) \lin_\varepsilon^\perp} \mathcal{P}_\perp \lin_\varepsilon \mathcal{P} W_v \d v  \d u\,,
\end{equation}
where we have used that 
\begin{equation*}
\mathcal{P}_{\perp} e^{u \mathcal{L}_{\varepsilon}} \mathcal{P} = \mathcal{P}_{\perp} (e^{u \mathcal{L}_{\varepsilon}} - e^{u \mathcal{L}_{\varepsilon}^{\perp}} ) \mathcal{P}=\int_{0}^{u}  \mathcal{P}_{\perp} \partial_{v} (e^{(u-v) \mathcal{L}_{\varepsilon}^{\perp}} e^{v \mathcal{L}_{\varepsilon}} ) \mathcal{P}  dv.
\end{equation*}
 As already outlined, Assumption \ref{ass:LLN} implies that  the constant $g$ defined in \eqref{Eq:g_def} is strictly positive, and, in particular, that the restriction of  $\lin_0$ to the range of $\mathcal{P}_\perp$  is bounded invertible, with an inverse bounded in norm by $ g^{-1}$.  Using a Neumann series expansion, using that 
 $g \geq 4 \varepsilon \|H\|$, we conclude that $\mathcal{L}_{\varepsilon}^{\perp}$ is bounded invertible, with 
 $$
\Big \| \frac{1}{\lin_\varepsilon^\perp} \Big \|_{2, \text{op}} \leq \frac{1}{g} \frac{1}{1 - 2 \varepsilon \|H \| / g } \leq \frac{2}{g} .
$$ 
 Exchanging the order of integrals in Equation \eqref{tointegrate}, we  get
\begin{align}
\label{WT}
\begin{split}
W_t = \mathcal{P}  &- \int_0^t \d u  \mathcal{P} \lin_\varepsilon \mathcal{P}_\perp (\lin_{\varepsilon}^{\perp})^{-1} \mathcal{P}_\perp \lin_\varepsilon \mathcal{P} W_u \\ 
&+ \underset{:=I_t}{\underbrace{\int_0^t \d u \mathcal{P} \lin_\varepsilon \mathcal{P}_\perp (\lin_{\varepsilon}^{\perp})^{-1} e^{(t-u) \lin_\varepsilon^\perp} \mathcal{P}_\perp \lin_\varepsilon \mathcal{P} W_u}}. 
\end{split}
\end{align}
By $I_t$  we denote the integral on the second line of \eqref{WT}. Using the Duhamel formula
$$
e^{t \lin_\varepsilon^\perp} = e^{t \lin_0^\perp} + \int_0^t e^{(t-u) \lin_0^\perp} (\lin_\varepsilon - \lin_0) e^{u \lin_\varepsilon^\perp} \d u,
$$
we get that
$$
e^{tg} \|e^{t \lin_\varepsilon^\perp} \mathcal{P}_\perp \|_{2, \text{op}} \leq 1 + 2 \varepsilon    \| H \| \int_0^t e^{ug} \|e^{u \lin_\varepsilon^\perp} \mathcal{P}_\perp\|_{2, \text{op}} \d u.
$$
Gr\"{o}nwall's inequality then implies that $\|e^{t \lin_\varepsilon^\perp} \mathcal{P}_\perp \|_{2, \text{op}} \leq e^{-t (g -2 \varepsilon \|H \|)},$ and hence, for \mbox{$g \geq 4 \varepsilon \|H \|$, we find that}
$$
\|I_t\|_{2, \text{op}} \leq 8 \frac{\varepsilon^2 \|H\|^2}{g} \int_0^t e^{-(t-u)(g -2 \varepsilon \|H\|)} \d u \leq 16\frac{\varepsilon^2 \|H\|^2}{g^2}.
$$
Next, we introduce two operators
\begin{align} \label{Qs}
Q_{\varepsilon}&:= -\varepsilon^{-2} \mathcal{P} \lin_\varepsilon \mathcal{P}_\perp (\lin_{\varepsilon}^{\perp})^{-1} \mathcal{P}_\perp \lin_\varepsilon \mathcal{P},\\
 \notag  Q&: = - \varepsilon^{-2} \mathcal{P} \lin_\varepsilon \mathcal{P}_\perp \lin_0^{-1} \mathcal{P}_\perp \lin_\varepsilon \mathcal{P}.
\end{align}
At the end of this subsection we will prove that $Q_{\varepsilon}$ is a dissipative operator.  It is left to the reader to verify that the action of $Q$ on any matrix $\rho$ can be expressed in terms of the matrix elements defined in \eqref{Qdef}, using that $\mathcal{P} \lin_\varepsilon \mathcal{P}_\perp =  -i \varepsilon \mathcal{P} \text{ad}_H \mathcal{P}_\perp$: 
\begin{equation*}
Q \rho = \sum_{\nu \in \sigma(\mathcal{N})} Q(\nu,\nu) P_{\nu} \rho_{\nu \nu} + \underset{\nu \neq \nu'  \in \sigma(\mathcal{N})}{\sum} Q(\nu,\nu') P_{\nu'}\rho_{\nu \nu}.
\end{equation*}
Differentiating Eq. \eqref{WT} in time $t$ yields
\begin{align}
\label{WT1}
\frac{ d }{ dt } W_t =  \varepsilon^2  Q_{\varepsilon}  W_t +  \frac{ d }{ dt } I_t .  
\end{align}
Using the method of variation of parameters we find that
\begin{align*}
W_{t} = e^{\varepsilon^2 t Q_\varepsilon} \mathcal{P} + e^{\varepsilon^2 t Q_\varepsilon} \mathcal{P}  \int_0^t 
  e^{-\varepsilon^2 u Q_\varepsilon} \mathcal{P} \frac{ d }{ du } I_u du .
\end{align*}
After an integration by parts we arrive at
$$
W_t = e^{\varepsilon^2 t Q_\varepsilon} \mathcal{P} + I_t + \varepsilon^2 Q_\varepsilon \int_0^t e^{\varepsilon^2 (t-u) Q_\varepsilon} I_u \d u.
$$
Putting $t= \varepsilon^{-2} s$, it follows that 
\begin{align*}
\| W_{\varepsilon^{-2} s}  - e^{s Q_\varepsilon} \mathcal{P}\|_{2, \text{op}}  &\leq 16\frac{\varepsilon^2 \|H\|^2}{g^2}\left(1  + 8 s \frac{\|H\|^2}{g} \right) \\
&\leq 16\frac{\varepsilon^2 \|H\|^2}{g^2} + 32 s \varepsilon \frac{\|H \|^3}{g^2}.
\end{align*}
To get {the first inequality we have used that $\|e^{s Q_{\varepsilon}}\|_{2,\text{op}} \leq 1$, because $Q_{\varepsilon}$ is dissipative, and, to get }the second inequality, we have used again that $g \geq 4 \varepsilon \|H \|$.
One further application of Duhamel's formula  yields
$$ \|e^{sQ_\varepsilon } - e^{sQ}\|_{2, \text{op}}  \leq 16 \varepsilon s \frac{\|H\|^3}{g^2}, $$ 
from which the last inequality claimed in the lemma follows after using the triangle inequality.

 It remains to prove that $Q_\varepsilon$ is a dissipative operator, see Eq. \eqref{dis}. This is a direct consequence of the fact that $ \lin_\varepsilon^\perp $ is dissipative and of the following calculation:
 $$
  \langle X,  Q_\varepsilon X \rangle = \langle Y,   (\lin_\varepsilon^\perp)^{-1}Y \rangle  =  \langle \lin_\varepsilon^\perp (\lin_\varepsilon^\perp)^{-1}  Y,   (\lin_\varepsilon^\perp)^{-1}Y \rangle    ,
 $$  
 for an arbitrary operator $X$, with $Y := \mathcal{P}_\perp([H,\mathcal{P}X])$.
\end{proof}
Theorem~\ref{thm:1} is an immediate corollary of the lemma. 

 Theorem~\ref{thm:1} is a representative of a growing literature devoted to studying various aspects of perturbation theory for Lindbladians; see, e.g., \cite{Macieszczak, Albert} and references therein.

\subsection{Comparison of  the full dynamics with the non-demolition dynamics}
We propose to compare the full dynamics ${\sigma}_\varepsilon^{(s,u)}$  with the non-demolition dynamics 
${\sigma}_0^{(s,u)}$. In order to get estimates that are tight enough to yield a proof of Theorem~\ref{thm:case1}, we have to expand the evolution operator ${\sigma}_\varepsilon^{(s,u)}$ to third order in $\varepsilon$. This renders the proof of the following lemma rather tedious. In order not to interrupt the flow of thought with technicalities of little interest, we relegate some details of the proof to Appendix~\ref{App:A}.
\begin{lemma}
\label{lem:5.4}
Let $E$ be a measurable subset of $\Xi$ and $s>u$. Then 
\begin{align}
\label{331}
&\|\mathcal{P}({\sigma}_\varepsilon^{(s,u)}(E) - {\sigma}_0^{(s,u)}(E)) \mathcal{P}\|_{2, \text{op}} \leq C\left[\varepsilon^2 (s-u) + \varepsilon^4 (s-u)^2 + \varepsilon^3 (s-u)^3\right],\\
\label{332}
&\| \mathcal{P}_\perp({\sigma}_\varepsilon^{(s,u)}(E) - {\sigma}_0^{(s,u)}(E)) \|_{2, \text{op}} \leq C\left[ \varepsilon + \varepsilon^2 (s-u)^2 \right], \\ 
 &\| ({\sigma}_\varepsilon^{(s,u)}(E) - {\sigma}_0^{(s,u)}(E)) \mathcal{P}_\perp\|_{2, \text{op}} \leq C\left[ \varepsilon + \varepsilon^2 (s-u)^2 \right], \label{333}
\end{align}
for some constant $C$ depending on $\|H\|, g$ and the dimension of $\mathcal{H}$.
\end{lemma}

\begin{proof}
The set $E$ is the disjoint union of two sets $E \cap \{ N_s - N_u < 3 \}$ and $E \cap \{ N_s - N_u \geq 3   \ \}   $. These two subsets require separate analysis, and, for the sake of brevity, we only consider the more important latter case. Hence, we assume that $E \cap \{ N_s - N_u < 3 \}$ is empty.
We start from the definition of the time evolution operator ${\sigma}_\varepsilon^{(s,u)}$ in \cref{prop}. A standard perturbative expansion yields
\begin{align*}
{\sigma}_\varepsilon^{(s,u)}( \underline{  \tau }, \underline \xi) =& {\sigma}_{0}^{(s,u)}( \underline{\tau}, \underline \xi) \\
										&+\sum_{j=N_{u}+1}^{N_{s}}    {\sigma}_{0}^{(s,t_{j})}( \underline{ \tau }, \underline \xi)  (B_{j} - \mathds{1}) {\sigma}_{\varepsilon}^{(t_{j},u)}  ( \underline{\tau}, \underline \xi),
\end{align*}
where $B_{N_s} =e^{-i \varepsilon (s-t_{N_s}) \adj_H},$  $B_{N_{u}+1}=  e^{-i \varepsilon (t_{N_{u}+1}-u) \adj_H}$, and $B_{j}=e^{-i \varepsilon  \tau_{j+1}  \adj_H }$,
otherwise. Notice that by definition (see \eqref{prop}), and the fact that $ N_{t_j } = j $, it follows that ${\sigma}_0^{(s, t_j)} =  \Phi_{\xi_{N_s}} \cdots   \Phi_{\xi_{j+1}} $ and $  {\sigma}_{\varepsilon}^{(t_j, u)} =    \Phi_{\xi_{j}} 
e^{ -i  \varepsilon \tau_j \ad_H } \cdots \Phi_{\xi_{N_{u} +1}}  e^{ -i  \varepsilon (t_{N_u +1} -u)  \ad_H  } $. 
Iterating this expansion step twice, we obtain
\begin{align*}
{\sigma}_\varepsilon^{(s,u)}( \underline{  \tau }, \underline \xi) - {\sigma}_{0}^{(s,u)}( \underline{  \tau }, \underline \xi) &= \sum_{j} A_{j}^{(1)}(\underline{ \tau  }, \underline{\xi}) + \sum_{j > k} R_{j,k}^{(2)}(\underline{ \tau  }, \underline{\xi}) \\
			&=\sum_{j} A_{j}^{(1)}(\underline{ \tau }, \underline{\xi}) + \sum_{j > k} A^{(2)}_{j,k}(\underline{  \tau }, \underline{\xi}) + \sum_{j > k > l} R_{j,k,l}^{(3)}(\underline{ \tau  }, \underline{\xi}),
\end{align*}
where all indices in the sums are constrained to the interval $(N_u, N_s]$, and the random variables appearing in the expansion are given by
\begin{align*}
A_{j}^{(1)} =&{\sigma}_{0}^{(s,t_{j})}  (B_{j} - \mathds{1}) {\sigma}_{0}^{(t_{j},u)}, \\
R_{j,k}^{(2)} =&{\sigma}_{0}^{(s,t_{j})}  (B_{j} - \mathds{1}) {\sigma}_{0}^{(t_{j},t_{k})}   (B_{k} - \mathds{1}) {\sigma}_{\varepsilon}^{(t_{k},u)},\\
A_{j,k}^{(2)} =&{\sigma}_{0}^{(s,t_{j})}  (B_{j} - \mathds{1}) {\sigma}_{0}^{(t_{j},t_{k})}   (B_{k} - \mathds{1}) {\sigma}_{0}^{(t_{k},u)},	\\
R_{j,k,l}^{(3)} =&{\sigma}_{0}^{(s,t_{j})}  (B_{j} - \mathds{1}) {\sigma}_{0}^{(t_{j},t_{k})}  (B_{k} - \mathds{1}) {\sigma}_{0}^{(t_{k},t_l)}	(B_{l} - \mathds{1}) {\sigma}_{\varepsilon}^{(t_{l},u)}.	
\end{align*}
Estimating the integrations of the above expressions over the set $E$ is quite a cumbersome task.  The main difficulty is that although we have a general estimate, $\|{\sigma}_\varepsilon^{(s,t)}(E)\|_{2,op} \leq 1$, see Lemma \ref{lem:CPDS}, there is no useful estimate on the integral $\int_E \|{\sigma}_\varepsilon^{(s,t)}(\underline{ {\tau}  }, \underline{\xi})\|_{2,op} \d \mathbb{P}(\underline{  {\tau} }, \underline{\xi}).$ The way to overcome this difficulty is to first bound traces of such integrals from above by expressions with positive integrands and then extend the integration to the whole space. Once this is accommplished, the integral factorizes and we can use the above estimates on ${\sigma}_\varepsilon^{(s,t)}(E)$. To obtain the right scaling with respect to $\varepsilon$ and $(s-u)$ these estimates have to be carried out differently for the $R$-terms and $A$-terms. Details are presented in Appendix~\ref{App:A}. The correct scaling with respect to 
$\varepsilon$ and $(s-u)$ can be inferred from a simple bookkeeping argument in which each $(B_j -\id)$ contributes by $\varepsilon \tau_{j+1} $, $\mathcal{P}(B_j -\id) \mathcal{P}$ contributes by $\varepsilon^2 (\tau_{j+1} )^2$, and $\mathcal{P}_\perp {\sigma}_{0}^{(t_j,t_k)}$ contributes by a geometric factor \,\,$(1-g)^{j-k}$.

We use the obvious bound 
(see Appendix \ref{App:0} - Eq. \eqref{ESta} - for the definitions of $\Omega$ and $P$) 
$$
\Big |\int_{E} f(\underline{ \tau }, \underline{\xi}) \Big | \d \mathbb{P}(\underline{  {\tau} }, \underline{\xi}) \leq \int_{\Omega} \Big |\int_{E_{\underline{   \tau }}} f(\underline{ \tau }, \underline{\xi}) \d \mu^{\otimes \mathbb{N}}(\underline{\xi}) \Big | \d P(\underline{  \tau  }),
$$
where $f(\underline{\tau}, \underline{\xi})$ is an arbitrary integrable function, and, for 
$\underline{\tau} \in [0,\infty)^{\mathbb{N}}$, we define $E_{\underline{\tau}  } = \{ \underline{\xi} : (\underline{ \tau}, \underline{\xi}) \in E \}$. In Lemma~\ref{lem:a2} we prove that there exists a constant $C$ such that 
\begin{align*}
\Big \|\int_{E_{\underline{ \tau }}} \mathcal{P} &({\sigma}_\varepsilon^{(s,u)}( \underline{ \tau  }, \underline \xi) - {\sigma}_{0}^{(s,u)}( \underline{ \tau }, \underline \xi) )\mathcal{P} \d \mu^{\otimes \mathbb{N}}(\underline{\xi})\Big \|_{2,op} \leq  C \Big( 
\sum_{j=N_u+1}^{N_s}  \varepsilon^2  (\tau_{j+1})^2 \\ +& \sum_{N_u+1 \leq k <j \leq N_s} \varepsilon^4 ( \tau_{j+1} )^2 (\tau_{k+1} )^2 +\varepsilon^2  \tau_{j+1} \tau_{k+1}  (1-g)^{j-k-1} \\
+& \sum_{N_{u}+1 \leq l < k <j  \leq N_s} \varepsilon^3  
    \tau_{j+1}  \tau_{k+1}   \tau_{l+1}  \Big).
\end{align*}
The lines correspond to bounds on $A_j^{(1)}, A_{j,k}^{(2)}$ and $R_{j,k,l}^{(3)}$. To estimate the expectation of this sum, we consider two cases. For $s - u \leq 1$, we estimate $\tau_{j+1} $ by $s-u$, and 
$\mathbb{E}[(N_s - N_u)^k]$ by a constant, for $k=1,2,3$. For $s - u >1$, the expectation value is estimated in Appendix~\ref{B1}.
In both cases we obtain
$$
\|\mathcal{P}({\sigma}_\varepsilon^{(s,u)}(E) - {\sigma}_0^{(s,u)}(E)) \mathcal{P}\|_{2, \text{op}} \leq C\left[\varepsilon^2 (s-u) + \varepsilon^4 (s-u)^2 + \varepsilon^3 (s-u)^3\right].
$$
This establishes Eq.~(\ref{331}). 

In Lemma~\ref{lem:a2} we show that there exists a constant $C$ such that
\begin{align*}
\Big \|\int_{E_{\underline{  {\tau} }}} \mathcal{P}_\perp ({\sigma}_\varepsilon^{(s,u)}( \underline{ {\tau} }, \underline \xi) - {\sigma}_{0}^{(s,u)}( \underline{ {\tau}  }, \underline \xi) ) \d \mu^{\otimes \mathbb{N}}(\underline{\xi})  \Big \|_{2,op} \leq& C \Big( \sum_{j=N_u+1}^{N_s}\varepsilon
  \tau_{j+1} (1-g)^{N_s - j -1} \\
& + \sum_{j,k=N_{u}+1}^{N_s} \varepsilon^2   \tau_{j+1} 
  \tau_{k+1}  \Big).
\end{align*}
Taking the expectation value over the (random-time) Poisson process we conclude that there exists constant $C$ such that (see Appendix~\ref{B1}) 
$$
\big \|\mathcal{P}_\perp({\sigma}_\varepsilon^{(s,u)}(E) - {\sigma}_0^{(s,u)}(E)) \big \|_{2, \text{op}} \leq C\left[\varepsilon + \varepsilon^2 (s-u)^2 \right].
$$
This yields \eqref{332}. The proof of Inequality \eqref{333} is similar and is omitted.
\end{proof}

\subsection{Convergence of finite-dimensional distributions to a Markov jump process}
Under the non-demolition dynamics ${\sigma}_{0}^{(s,u)}$, the state $\rho$ of the system asymptotically purifies to an eigenstate of $\mathcal{N}$.  Choosing a sampling time $T=T(\varepsilon)$ and using the definitions in Section \ref{loglike}, we introduce sets of maximum likelihood
\begin{equation}
\label{Enu}
E_{\nu}^{(T+s,s)}:=\{ \n_s(T) = \nu\},
\end{equation}
and we use the symbol $ \E_{\mu_{\nu}} $ to denote an expected value with respect to 
$  \mu_\nu $.

Let $X_{\nu,\nu'}(\xi)$ be given by  $X_{\nu, \nu'}(\xi)=\log \big(f(\xi | \nu)/ f(\xi | \nu')\big)$. Jensen's inequality, applied to the convex function $r \mapsto - r^a  $, with  $ r =  f(\xi | \nu)/ f(\xi | \nu') $ and $ 0 \leq a \leq 1$, implies that
$$-\log \E_{\mu_{\nu'}}[\exp(a X_{\nu, \nu'})]\geq 0,
$$  
and, by Assumption~{\ref{ass:LLN}}.2, it is strictly positive, for $a >0$.
 We define
\begin{equation}
\label{def_I}
I := \sup_{0 \leq a \leq 1} \underset{\nu \neq \nu'}{\min} \left( -\log \E_{\mu_{\nu'}}[\exp(a X_{\nu, \nu'})] \right).
\end{equation}
Then $I$ is a strictly positive number.

Recall the definitions of $\caP, \caP_\perp$ from the beginning of \Cref{dyave}.
\begin{lemma}
\label{lemND}
Let $T>0$. Under Assumption \ref{ass:LLN}, there exists a constant $C$ such that the  inequalities  
\begin{align}
\label{eq:c_type}
&\|\mathcal{P} {\sigma}_0^{(T+s,s)}(E_{\nu}^{(T+s,s)}) - \mathcal{P}_\nu\|_{2,\text{op}} \leq C  e^{-T(1-e^{-I})},\\
\label{eq:d_type}
& \|\mathcal{P}_\perp {\sigma}_0^{(T+s,s)}(E_{\nu}^{(T+s,s)})\|_{2, \text{op}} \leq  C e^{-T(1-e^{-I})/2},
\end{align}
hold for an arbitrary eigenvalue $\nu$ of the observable $\caN$.
\end{lemma}

\begin{proof}
If $(\underline{\tau}, \underline \xi)$ belongs to the set  $E_{\nu}^{(T+s,s)}$, then $\nu$ maximizes the log-likelihood ratio.   Hence 
$$
\sum_{j={N_s+1}}^{N_{s+T}} \log f (\xi_j | \nu) \geq \sum_{j={N_s+1}}^{N_{s+T}} \log f (\xi_j | \nu{'}),
$$
for all $\nu' \in \sigma(\mathcal{N})$. From Markov's inequality, we deduce that the inequality
$$
\mathbb{P}^0_{P_{\nu'}}(E_{\nu}^{(T+s,s)}) \leq \mathbb{E}^0_{P_{\nu'}} \left[ \exp (a  \sum_{j={N_s+1}}^{ N_{s + T} } (\log f (\xi_j | \nu) -\log f (\xi_j | \nu')))\right]
$$
 holds for any $0 \leq a $.  From the definition of the measure  $  \mathbb{P}^0_{P_{\nu'}} $ (see also Appendix \ref{App:0}) we get 
$$
\mathbb{P}^0_{P_{\nu'}}(E_{\nu}^{(T+s,s)}) \leq \mathbb{E}\left[ (\mathbf{E}_{\mu_{\nu'}}
[ \exp(a X_{\nu,\nu'})])^{ N_{s + T} - N_{s}} \right].
$$
Using Eq. \eqref{def_I} and the identity $\mathbb{E}[b^{N_s}] = e^{-(1-b)s}$, which is valid for any positive real number $b$, we conclude that (we also use stationarity and independence of increments for Poisson processes)
$$
\mathbb{P}^0_{P_{\nu'}}(E_{\nu}^{(T+s,s)}) \leq  e^{-T(1-e^{-I})} 
$$
holds for all $\nu \neq \nu'$. Since $\mathcal{P}_{\nu'} {\sigma}_0^{(T+s,s)}(E_{\nu}^{(T+s,s)}) = \mathbb{P}^0_{P_{\nu'}}(E_{\nu}^{(T+s,s)}) \mathcal{P}_{\nu'}$, this implies that
\begin{align*}
\|\mathcal{P} {\sigma}_0^{(T+s,s)}(E_{\nu}^{(T+s,s)}) - \mathcal{P}_\nu \|_{2,\text{op}} &\leq |1- \mathbb{P}^0_{P_{\nu}}(E_{\nu}^{(T+s,s)})| + \sum_{\nu' \neq \nu} \mathbb{P}^0_{P_{\nu'}}(E_{\nu}^{(T+s,s)})\\
 &\leq 2\dim(\mathcal{H}) e^{-T(1-e^{-I})}.
\end{align*}
Inequality \eqref{eq:d_type} of the lemma is established in a similar manner.  See \cite{BFFS} for related arguments.
\end{proof}

Combining Lemma~\ref{lem:5.1}, estimates (\ref{eq:c_type}), \eqref{eq:d_type}, and Lemma~\ref{lem:5.4}, we obtain an expansion of the evolution operator to leading order, conditioned on the event that  $\n_t=\nu$. We constrain the values of the sampling time $T$ such as to get upper bounds that are all of order $\varepsilon$.

\begin{lemma}
\label{lem:5.5}
Suppose that $g > 4 \varepsilon \|H\|$ and $T \in  \big ( - \big (\log(\varepsilon) \big ) 2    (1 - e^{-I})^{-1}, \varepsilon^{-\frac{1}{2}}  \big )$. There is a constant $C$ such that, for any $s \geq 0$ and any $u \in ( -\log( \varepsilon) \varepsilon^2 g^{-1},1)$,
\begin{equation}
\label{premb}
\| \mathcal{P} e^{\varepsilon^{-2} u \lin_\varepsilon} {\sigma}_\varepsilon^{(T+s,s)}(E_{\nu}^{(T+s,s)}) \mathcal{P} - e^{uQ} \mathcal{P}_\nu \|_{2, \text{op}} \leq  C \varepsilon
\end{equation}
and
\begin{equation} \label{0071}
\|e^{\varepsilon^{-2} u\lin_\varepsilon} {\sigma}_\varepsilon^{(T+s,s)}(E_{\nu}^{(T+s,s)})  \mathcal{P}_\perp \|_{2,\text{op}} \leq C\varepsilon,
\end{equation}
\begin{equation}
\label{007}
\|\mathcal{P}_\perp e^{\varepsilon^{-2} u \lin_\varepsilon} {\sigma}_\varepsilon^{(T+s,s)}(E_{\nu}^{(T+s,s)}) \|_{2, \text{op}} \leq  C \varepsilon . 
\end{equation}
\end{lemma}

\begin{proof}
The assumption that $g > 4 \varepsilon \|H\|$  allows us to apply Lemma~\ref{lem:5.1}. The  bounds are  obtained by standard telescoping. To ease notations, let $\mathcal{T}_{\varepsilon}:={\sigma}_\varepsilon^{(T+s,s)}(E_{\nu}^{(T+s,s)})$; (the explicit dependence on $T$, $\nu$ and $s$ is omitted). We use repeatedly that $||\mathcal{T}_\varepsilon||_{2,op} \leq 1$, see Lemma \ref{lem:CPDS}.
To prove the first bound, we use that 
\begin{align*}
 \mathcal{P} e^{\varepsilon^{-2} u \lin_\varepsilon} \mathcal{T}_{\varepsilon} \mathcal{P} - e^{uQ} \mathcal{P}_{\nu} &= \mathcal{P} e^{\varepsilon^{-2} u \lin_\varepsilon} (\mathcal{P} + \mathcal{P}_{\perp}) \mathcal{T}_{\varepsilon} \mathcal{P} - e^{uQ} \mathcal{P}_{\nu}\\
 &=  \mathcal{P} e^{\varepsilon^{-2} u \lin_\varepsilon} \mathcal{P} \mathcal{T}_{\varepsilon} \mathcal{P} - e^{uQ} \mathcal{P}_{\nu} + R_1
\end{align*}
where $R_1:= \mathcal{P} e^{\varepsilon^{-2} u \lin_\varepsilon}  \mathcal{P}_{\perp} \mathcal{T}_{\varepsilon} \mathcal{P}$ is bounded in norm by $e^{-\varepsilon^{-2} u g} + \frac{2 \varepsilon \|H \|}{g}$, as follows from \eqref{ineq11}. Since $u$ is {larger than} $- \varepsilon^2\log(\varepsilon)  g^{-1}$, $R_1$ is bounded in norm  by a term of order $\varepsilon$. Similarly, introducing 
$R_2:=\mathcal{P} e^{\varepsilon^{-2} u \lin_\varepsilon} \mathcal{P} (\mathcal{T}_{\varepsilon}-\mathcal{T}_0)  \mathcal{P}, $
we have that 
\begin{align*}
\mathcal{P}  e^{\varepsilon^{-2} u \lin_\varepsilon} \mathcal{P}  \mathcal{T}_{\varepsilon} \mathcal{P} - e^{uQ} \mathcal{P}{_\nu} &= \mathcal{P} e^{\varepsilon^{-2} u \lin_\varepsilon} \mathcal{P} \mathcal{T}_0 \mathcal{P} - e^{uQ} \mathcal{P}_{\nu} + R_2\\
&= (\mathcal{P} e^{\varepsilon^{-2} u \lin_\varepsilon} \mathcal{P}- e^{uQ}) \mathcal{P} \mathcal{T}_0 \mathcal{P} +  e^{uQ} \mathcal{P} (\mathcal{P} \mathcal{T}_0 \mathcal{P} - \mathcal{P}_{\nu}) + R_2 .
\end{align*}
Using Inequality \eqref{ineq12} in Lemma \ref{lem:5.1} and the fact that $u \in (0,1)$, the first term in the above equation is seen to be norm-bounded by a term of order $\varepsilon$.  The second term is norm-bounded by a constant of order $\varepsilon$, as is seen by using \eqref{eq:c_type} in Lemma \ref{lemND} and our assumption on the value of $T$. The Term $R_2$ is norm-bounded by a constant of order $\varepsilon$, as well, which follows from \eqref{331} in Lemma \ref{lem:5.4} and our assumption that $T$ is smaller than $\varepsilon^{-1/2}$. This establishes Eq. \eqref{premb}. The second inequality, Eq. \eqref{0071}, is proven by using the estimate
\beq
\|e^{\varepsilon^{-2} u \lin_\varepsilon} \mathcal{T}_{\varepsilon}  \mathcal{P}_\perp \|_{2, \text{op}}  \leq \|\mathcal{T}_0 \mathcal{P}_\perp \|_{2,\text{op}} + \| (\mathcal{T}_{\varepsilon} - \mathcal{T}_0) \mathcal{P}_\perp \|_{2, \text{op}}
\eeq
and by bounding the terms on the RHS by \eqref{eq:d_type}  and \eqref{333}.
The last part of the lemma, follows from Lemma~\ref{lem:5.1}: 
\begin{align*}
\|\mathcal{P}_\perp e^{\varepsilon^{-2} u \lin_\varepsilon} \mathcal{T}_{\varepsilon} \|_{2, \text{op}} \leq \|\mathcal{P}_\perp e^{\varepsilon^{-2} u \lin_\varepsilon}\|_{2, \text{op}}  &\leq e^{-\varepsilon^{-2} u g} + \varepsilon \frac{2 \|H\|}{g}.
\end{align*}  
The condition on $u$ ensures that the term $e^{-u \varepsilon^2 g}$ is of order $\varepsilon$. 
\end{proof}

\begin{lemma}
\label{lem:5.7}
Suppose that $g > 4 \varepsilon \|H \|$, $T \in \big ( - \big (\log(\varepsilon) \big ) 2    (1 - e^{-I})^{-1}, \varepsilon^{-\frac{1}{2}}  \big )$, and
let $(s_j)_{j=1}^{n+1}$ be a strictly increasing sequence of times in the interval $ (0,1)$.
 Then there exists a constant $C$ independent of $\varepsilon$ and $n$ (but it depends on $s_1$) such that, for sufficiently small $\varepsilon$ (depending on the sequence $(s_j)_{j=1}^{n+1}$), 
\begin{align}
\label{48}
\Big |\mathbb{P}_\rho^\varepsilon(\{ \mathcal{M}_{s_j \varepsilon^{-2}}  = \nu_j : j=1, \dots, n\})
 - \mathbb{P}_{\pi_{\rho}}(\{Y_{s_j} = \nu_j, j =1,\dots, n\}) \Big | & \leq C n \varepsilon, 
\end{align}
and 
\begin{align}\label{48prima} 
 \Big \| {\sigma}_\varepsilon^{(\varepsilon^{-2} s_{n + 1},\; 0)} (\cap_{j=1}^n \{ \mathcal{M}_{s_j \varepsilon^{-2}} = \nu_j\}) -  e^{ (s_{n + 1} - s_n) Q } \mathcal{P}_{\nu_n}    \dots  e^{ (s_2 - s_1) Q } \mathcal{P}_{\nu_1} e^{s_1 Q}  \Big \|_{2, \text{op}} &\leq C n \varepsilon.
\end{align}
\end{lemma} 
\vspace{2mm}

\begin{proof} 
   We recall that $  \mathcal{M}_{s} $ is defined to be $\n_{m T}(T) $, whenever $ s \in [mT,  (m+1)T) $, see Eq. \eqref{proc}, for some $ m \in \mathbb{N}$. Let $r_j, j\in \{1 \dots,n\}$ be positive real numbers such that $r_j \varepsilon^{-2}$ is a multiple of $T$ and $ s_j \varepsilon^{-2}  \in [ r_j \varepsilon^{-2},    r_j \varepsilon^{-2} + T)$. Then we have that $\mathcal{M}_{s_j \varepsilon^{-2}  } =  \n_{  r_j \varepsilon^{-2}  }(T)$. For a fixed sequence $(s_j)_{j=1}^{n+1}$, we choose $\varepsilon$ small enough so that $  \big \{ [   r_j  , r_j +  \varepsilon^{2} T)     \big \}_{  j \in \{1, \cdots, n \}  }   $ are disjoint.

We define sets  $E_{j}:=\{ \n_{ \varepsilon^{-2} r_j }(T)=\nu_j \}$.
To ease notations further, we use the abbreviation 
 $$\mathcal{T}_j:={\sigma}_\varepsilon^{( \varepsilon^{-2} r_j + T, \;   \varepsilon^{-2} r_j )}(E_{j}).$$
Using \Cref{prop:unraveling} and Eqs.~(\ref{factorization}), \eqref{pgp}, we have that 
\begin{align}
 {\sigma}_\varepsilon^{(\varepsilon^{-2} s_{n + 1},\; 0)}  
 \Big (\cap_{j=1}^n \{ \n_{   \varepsilon^{-2} r_j }(T)   =    \nu_j\} \Big )  = &
 {\sigma}_\varepsilon^{(\varepsilon^{-2} s_{n + 1},\; 0)}  \Big (\cap_{j=1}^n \{ \mathcal{M}_{ s_j \varepsilon^{-2}}  =    \nu_j\} \Big ) 
\\ \notag  = & e^{  \varepsilon^{-2} \big ( s_{n + 1} - (r_{n} +  \varepsilon^{2} T)\big ) \lin_{\varepsilon}}     
\mathcal{T}_n  e^{   \varepsilon^{-2} \big ( r_{n}  - (r_{n-1} +   \varepsilon^{2} T)\big ) \lin_{\varepsilon}}     
 \\ \notag & \mathcal{T}_{n-1} \dotsc   e^{ \varepsilon^{-2} \big ( r_{2}  - (r_{1} +  \varepsilon^{-2} T)\big )  \lin_{\varepsilon}}   
\mathcal{T}_1 e^{  \varepsilon^{-2} r_1 \lin_{\varepsilon}}.
\end{align}
Inserting $\mathcal{P} + \mathcal{P}_{\perp}=\mathds{1}$ in front of each operator of the form $e^{ \varepsilon^{-2}( \cdot ) \lin_{\varepsilon}}$ and to the right of $e^{   \varepsilon^{-2}  r_1 \lin_{\varepsilon}}$, we obtain that
\begin{align}
\label{splitted}
{\sigma}_\varepsilon^{(\varepsilon^{-2} s_{n + 1}, \; 0)}   (
 \cap_{j=1}^n \{ \mathcal{M}_{s_j \varepsilon^{-2}}  = \nu_j\})= & \mathcal{P} 
  e^{  \varepsilon^{-2}  \big (  s_{n + 1} - (r_{n} +   \varepsilon^{2} T)\big ) \lin_{\varepsilon}}     
\mathcal{T}_n \mathcal{P}   e^{  \varepsilon^{-2} \big ( r_{n}  - (r_{n-1} +   \varepsilon^{2} T)\big ) \lin_{\varepsilon}} \nonumber \\    
& \mathcal{T}_{n-1}  \cdots \mathcal{P} 
 e^{  \varepsilon^{-2} \big ( r_{2}  - (r_{1} +   \varepsilon^{2} T)\big )  \lin_{\varepsilon}}   
\mathcal{T}_1  \mathcal{P}   e^{  \varepsilon^{-2}  r_1 \lin_{\varepsilon}} \mathcal{P}+ R, 
\end{align}
where, using Lemma \ref{lem:CPDS}, the remainder $R$ is bounded by 
\begin{align}\label{sip1}
\big \| \mathcal{P}_{\perp} &    e^{  \varepsilon^{-2}  \big ( s_{n + 1} - (r_{n} +   \varepsilon^{2} T)\big ) \lin_{\varepsilon}}      
\mathcal{T}_n \big \|_{2, \text{op}} +  \big \|  \mathcal{P}_{\perp}   e^{   \varepsilon^{-2} \big ( r_{n}  - (r_{n-1} +  \varepsilon^{-2}  T)\big ) \lin_{\varepsilon}}     
\mathcal{T}_{n-1} \big \|_{2, \text{op}} +    \cdots \nonumber \\  & \cdots +  \big \| \mathcal{P}_{\perp} 
 e^{  \varepsilon^{-2}  \big ( r_{2}  - (r_{1} +  \varepsilon^{2}  T)\big )  \lin_{\varepsilon}}   \mathcal{T}_1   \big \|_{2, \text{op}} +  \big \| \mathcal{P}_{\perp}  e^{  \varepsilon^{-2}  r_1 \lin_{\varepsilon}} \big \|_{2, \text{op}}  + \big \|  e^{   \varepsilon^{-2}  r_1 \lin_{\varepsilon}} \mathcal{P}_{\perp} \big \|_{2, \text{op}}.
\end{align}
 By Eqs. \eqref{007} and \eqref{ineq11}, this is bounded by 
$n  C \varepsilon$, for some constant $C>0$ independent of $\varepsilon$.\footnote{provided we chose $\varepsilon $ small enough in order to fulfill the conditions for \cref{007}  and get the desired bound from \cref{ineq11} - notice that as $\varepsilon$ tends to zero, $r_j$ tends to $s_j$.} Next we notice that, for small enough $\varepsilon $, Eq. \eqref{premb} - Lemma  \ref{lem:5.5} - implies that 
\begin{align}\label{sip2}
\Big \| \mathcal{P} e^{  \varepsilon^{-2}(r_{j+ 1} - (r_j+   \varepsilon^{2} T)) \lin_{\varepsilon}}  \mathcal{T}_j \mathcal{P}   
  -e^{  (r_{j+ 1} - (r_j +  \varepsilon^{2}  T))   Q} \mathcal{P}_{\nu}  \Big \|_{2, \text{op}}  \leq C \varepsilon ,  
 \end{align}
  for each $j \in \{1, \cdots, n \}$.  
   Eq. \eqref{48prima} then follows  form \eqref{splitted}, the bounds in \eqref{sip1}, \eqref{sip2},  Lemma \ref{lem:5.1} - Eq. \eqref{ineq12} - and the fact that $| r_j - s_j| \leq C \varepsilon^2    $, for a finite constant $ C $ (that depends on $s_1$, see  Eq. \eqref{ineq12}). Finally, inequality \eqref{48}  is a consequence of  \eqref{48prima}, because the probabilities in \eqref{48} coincide with the traces of the operators inside the norm in \eqref{48prima}. Since the Hilbert space of the system is finite dimensional,  the trace-  and the Hilbert Schmidt norms are equivalent.     
\end{proof}

\subsection{Proof of part (a) of Theorem \ref{thm:case1}}
\label{proofa}
We now set out to prove the first part of our main theorem. To prove convergence in law of $\mathcal{M}_{\varepsilon^{-2}s}$ to $Y_s$, we use standard results from the theory of  convergence of stochastic processes; see e.g. \cite{Billingsley}.  The stochastic process $\mathcal{M}_{\varepsilon^{-2}s}$  in Theorem \ref{thm:case1} takes values in $\sigma(\mathcal{N})$, which is a discrete subset of $\mathbb{R}$; see Assumption \ref{ass:LLN}. Moreover, the variable $s$ is restricted to the interval $(0,1]$.  To prove convergence of $\mathcal{M}_{\varepsilon^{-2}s}$, as $\varepsilon \rightarrow 0$, we  make use of a standard convergence criterion stated here as Theorem \ref{Billin}.  It is  an immediate corollary of Theorem 13.3 in \cite{Billingsley}.\\
In he following, convergence in distribution is indicated by the double arrow "$\Rightarrow$".

\begin{theorem}
\label{Billin}
Let $(X_n)$ be a sequence of $\mathbb{R}$-valued stochastic c\`{a}dl\`{a}g processes, each of them defined on a measurable space ($\Omega_n,\mathcal{F}_n,\mathds{P}_n$), for times $s \in [0,1]$. Let $X$ be a $\mathbb{R}$-valued stochastic c\`{a}dl\`{a}g process defined on ($\Omega,\mathcal{F},\mathds{P}$), for times $s \in [0,1]$.   Suppose that 
\begin{enumerate}[(i)]
\item  the finite-dimensional  distributions of $X_n$  converge to the finite-dimensional distributions of  $X$, indicated as $X_n \overset{fd}{\Rightarrow_n} X$;
\item the limiting process $X$ satisfies $X(1)- X(1-\delta) \Rightarrow_{\delta \rightarrow 0} 0$;
\item for any $\eta, \delta>0$, there is $\theta \in (0,1)$ and $n_0$ such that, $\forall  n \geq n_0$,
\begin{equation}
\label{jumps}
\qquad  P_n (\{ \underset{\underset{t_2-t_1 \leq \theta}{t_1 \leq t \leq t_2}}{\sup} \{ \vert X_n(t)- X_n(t_1) \vert \wedge   \vert X_n(t_2)- X_n(t) \vert \} \geq \delta \} ) \leq \eta
\end{equation} 
where the supremum ranges over all triples of times $t,t_1,t_2$ in [0,1] satisfying the constraints $t_1<t<t_2$ and $t_2 - t_1 < \theta$.
\end{enumerate}
Then $X_n \Rightarrow X$.
\end{theorem}

\begin{proof}[Proof of part (a) of Theorem \ref{thm:case1}]
We must check the conditions of Theorem \ref{Billin} for the sequence of processes 
$X_{n}:=\mathcal{M}_{\varepsilon_n^{-2}s}$, for some sequence $(\varepsilon_{n})_{n=0}^{\infty}$ converging to $0$, with the role of $X$ being played by the Markov chain $Y_s$; see Eq. \eqref{Qdef}. Note that (ii) is automatically satisfied by the stochastic process $Y_s$, because $Y_s$ is a continuous Markov chain. Also, (i) follows directly from Lemma \ref{lem:5.7}.    Hence, to prove Theorem \ref{thm:case1} (a), it suffices to verify condition (iii) in Theorem 
\ref{Billin}. 

The stochastic process $\mathcal{M}_{\varepsilon^{-2}s}$ has piece-wise continuous paths and $\vert \mathcal{M}_{\varepsilon^{-2}s}- \mathcal{M}_{\varepsilon^{-2}u} \vert$ takes values in the finite set $\{ \vert \nu_j- \nu_k \vert \mid \nu_{j},\nu_{k } \in \sigma(\mathcal{N})\}$. In this special case, note that spectrum of $\mathcal{N}$ is non-degenerate by assumption, the probability in \eqref{jumps} is independent of $\delta$ for $\delta$ small enough. Hence \eqref{jumps} is proven once we have shown that there exist $\varepsilon_0$ such that for all $\varepsilon \leq \varepsilon_0$ the probability of paths that have two or more jumps in any interval of size less than $\theta$ goes to zero as $\theta$ goes to zero. We are going to show that this probability goes to zero linearly in $\theta$.

Let $F  \in {\mathcal{F}}$ be the event that  $\mathcal{M}_{\varepsilon^{-2} s}$ changes its value  at least  twice in the interval $(s_i,s_f)$ (of length $\theta$). For fixed $\varepsilon, T>0$, let $({u}_\ell)_{\ell=1}^N:=(\ell T)_{\ell=J_0+1}^{N+J_0}$,  be the integer multiples of $T$ in $\varepsilon^{-2}(s_i, s_f)$.  Let us decompose  $F$ according to the events that $\mathcal M_{{\varepsilon}^{-2}s}$ has jumps at times ${s}_j< {s}_k$, with $  \varepsilon^{-2} s_j = {u}_j  $ and  $  \varepsilon^{-2} s_k = {u}_k  $: set $   P_{j, k}(\nu_1, \nu_2,\nu_3,\nu_4)   $ the probability that  $  \mathcal{M}_{r} =  \nu_{1} $  (for $  r \in
[{u}_{k}, {u}_{k+1})  $),  $  \mathcal{M}_{r} =  \nu_{2} $ (for $  r \in
[{u}_{k-1}, {u}_{k})  $),  $  \mathcal{M}_{r} =  \nu_{3} $ (for $  r \in
[{u}_{j}, {u}_{j+1})  $) and  $  \mathcal{M}_{r} =  \nu_{4} $ (for $  r \in
[{u}_{j-1}, {u}_{j})  $), then we have that        
\beq
\label{52}
\mathbb{P}_\rho^{\varepsilon} (F)  \leq \sum_{j<k} \quad \sum_{\underset{\nu_3 \neq \nu_4}{\nu_1 \neq \nu_2} } P_{j, k}(\nu_1, \nu_2,\nu_3,\nu_4).
\eeq 
We define  $F_{\nu}^{({r})}:=\{ \mathcal{M}_{ r} = \nu\}.$ It follows that (see Eq.  \eqref{factorization})
\begin{align*}
P_{j, k}(\nu_1, \nu_2,\nu_3,\nu_4)= \tr \Big(  & e^{   (\varepsilon^{-2} s_\mathrm{f} -  u_{k+1})\lin_\varepsilon}   {\sigma}^{( u_{k+1}, 
  u_k)}_\varepsilon ( F_{\nu_1}^{(  {u}_k)})  {\sigma}^{({{u}}_{k},{{u}}_{k-1})}_\varepsilon(F_{\ \nu_2}^{({{u}}_{k-1})}) \\
& e^{   ({{u}}_{k-1}-{{u}}_{j+1} )\lin_\varepsilon}    {\sigma}^{({{u}}_{j+1},{{u}}_{j})}_\varepsilon(F_{\nu_3}^{({{u}}_j)})    {\sigma}^{({{u}}_{j},{{u}}_{j-1})}_\varepsilon(F_{\nu_4}^{({{u}}_{j-1})})  e^{ {{u}}_{j-1} \lin_\varepsilon} \rho  \Big).
\end{align*}
 Due to the symmetry of $P_{j, k}(\nu_1, \nu_2,\nu_3,\nu_4)$, it is sufficient to bound in norm a term of the form 
$e^{  {r} \lin_\varepsilon}  {\sigma}^{({{u}}_{k+1},{{u}}_k)}_\varepsilon ( F_{\nu_1}^{({{u}}_k)})  {\sigma}^{({{u}}_{k},{{u}}_{k-1})}_\varepsilon(F_{\ \nu_2}^{({{u}}_{k-1})}) e^{   \varepsilon^{-2} {r}' \lin_\varepsilon} $, for positive numbers  $ r, r' $, where  
$ {r} =    \varepsilon^{-2} s_\mathrm{f} -{{u}}_{k+1} $ and $ {r}' = \frac{1}{2} (  {{u}}_{k - 1 } - {{u}}_{j+1}) $ (changing $k$ by $j$ in the previous term would require that we take ${r} = \frac{1}{2} ( {u}_{k - 1 } - {{u}}_{j+1} )  $ and ${r}' =    {{u}}_{j-1}$). 
We insert the decomposition of the identity $\mathcal{P} + \mathcal{P}_\perp=1$ in between each of the indicated factors  and bound each resulting term separately. We obtain eight terms, denoted $(\beta_1, \beta_2, \beta_3)$ with $\beta_i\in \{\varnothing, \perp\}$, given by
\beq
(\beta_1, \beta_2, \beta_3):= \quad   e^{  r \lin_\varepsilon} \mathcal{P}_{\beta_1} {\sigma}^{({{u}}_{k+1},{{u}}_k)}_\varepsilon ( F_{\nu_1}^{({{u}}_k)}) \mathcal{P}_{\beta_2} {\sigma}^{({{u}}_{k},{{u}}_{k-1})}_\varepsilon(F_{\ \nu_2}^{({{u}}_{k-1})})  \mathcal{P}_{\beta_3} e^{  r' \lin_\varepsilon}, 
\eeq
\noindent  where $   \mathcal{P}_{\varnothing} \equiv \mathcal{P}  $   . We choose $\alpha_0$ such that the inequalities $\alpha_0 > 2(1-e^{-I})^{-1}$ and $\alpha_0 > g^{-1}$ are satisfied.
 By Lemmas~\ref{lem:5.1}, \ref{lem:5.4} and \ref{lemND} and our choice of $T$, there is a constant $C >0$ such that the following bounds hold (for every $a \geq \frac{1}{2} T $):
\begin{align}
\label{ia}
&\|  e^{  a \lin_\varepsilon} \mathcal{P}_{\perp}  \|_{2,\text{op} } \leq C \varepsilon,\\
\label{ib}
& \|  \mathcal{P}_{\perp} {\sigma}^{({{u}}_{k+1},{{u}}_k)}_\varepsilon ( F_{\nu_1}^{({{u}}_k)}) \|_{2,\text{op} } \leq C \varepsilon , \hspace{.5cm}  
   \|   {\sigma}^{({{u}}_{k+1},{{u}}_k)}_\varepsilon ( F_{\nu_1}^{({{u}}_k)})  \mathcal{P}_{\perp}  \|_{2,\text{op} } \leq C \varepsilon,     \\
\label{ic}
&\|\mathcal{P} {\sigma}^{({{u}}_{k+1},{{u}}_k)}_0 ( F_{\nu_1}^{({{u}}_k)}) \mathcal{P} {\sigma}^{({{u}}_{k},{{u}}_{k-1})}_0(F_{\ \nu_2}^{({{u}}_{k-1})})  \mathcal{P}\|_{2,\text{op}} \leq C \varepsilon^{2}  \quad (\nu_2 \neq \nu_1).
\end{align}
 It might happen that both $r$ and $r'$ are smaller that $T/2$, in the case that $  {u}_{k-1} = {u}_{j+1}  $, and it is also possible that  
$  {u}_{k} = {u}_{j+1}   $, but the number of these occurrences is relatively small and therefore these terms are easy to handle. Below we analyze the other terms.     
Using \eqref{ic} and \eqref{331}, the term $(\varnothing, \varnothing, \varnothing)$ is bounded in norm by $C \varepsilon^2 T$.  Except when one of $r$ or $r'$ is smaller than $T$ and $(\beta_1, \beta_2, \beta_3)\in \{(\perp, \varnothing, \varnothing), (\varnothing, \varnothing,\perp)\}$, the norm of $(\beta_1, \beta_2, \beta_3)$ is bounded by $C{\varepsilon}^2$, in all remaining cases. 

In the special case that $r$ or ${r}'$ is smaller than $T$,  and $(\beta_1, \beta_2, \beta_3)\in \{(\perp, \varnothing, \varnothing),\\ (\varnothing, \varnothing,\perp)\}$, the norms are bounded by $C\varepsilon$.  Such terms appear $\propto N$ times in \eqref{52}. Using  \eqref{52} and the above bounds, we deduce that 
\begin{equation}
\mathbb{P}_\rho^{\varepsilon} (F)  \leq C  N^2  \varepsilon^4 T^2 + CN  \varepsilon^2.
\end{equation}
Since $N^2 C T^2 \varepsilon^4 \sim C \theta^2$ and $N C  \varepsilon^2 \sim C \theta/T$ , we arrive at the bound
 $$
 \mathbb{P}_\rho^{(\varepsilon)}(F) \leq C \theta^2 + C \theta/ T .
 $$
Choosing $\theta$ small enough we conclude that \eqref{jumps} holds. 
 \end{proof}
\subsection{Proof of part (b) of Theorem \ref{thm:case1}}
\label{proofb}
\begin{proof}
The condition $s > 2 \varepsilon^2 T$ is necessary for the validity of the theorem; the projection $ P_{\hat{\mathcal{N}_0}}$ does not estimate the initial state $\rho_0$. In fact $\hat{\mathcal{N}}_0$ uses measurement  results in the interval $[0,T)$ and $ P_{\hat{\mathcal{N}_0}}$ is a good estimate of $\rho_T$. We employ this observation in the proof by shifting the argument by $T$.

For $s > 2 \varepsilon^2 T$, we estimate 
$$
{\mathbb{E}^{\varepsilon}_\rho} \Big [ \|\rho_{\varepsilon^{-2}s}-P_{\hat{\mathcal {N}}_{\varepsilon^{-2}s}}   \|_2  \Big ] \leq 
  {\mathbb{E}^{\varepsilon}_\rho}   \Big [ \|\rho_{\varepsilon^{-2}s}-P_{\hat{\mathcal {N}}_{\varepsilon^{-2}s-T}}\|_2 \Big ] + {\mathbb{E}^{\varepsilon}_\rho} \Big [\| P_{\hat{\mathcal {N}}_{\varepsilon^{-2}s}} -  P_{\hat{\mathcal {N}}_{\varepsilon^{-2}s-T}}\|_2 \Big ].
$$
The second term on the right side is bounded by twice the probability of making a jump in the interval of length $T$, in the proof of part (a) we estimated that this probability is of order $\varepsilon^2 T$.
  Notice that $\varepsilon^2 T$ is dominated by $\varepsilon$ for the choice $T = \alpha \log \varepsilon$.

To bound the first term on the right side, we note that
\begin{align*}
\|\rho_{\varepsilon^{-2}s}-P_{\hat{\mathcal {N}}_{\varepsilon^{-2}s-T}}\|_2^2 &= \tr(\rho^2_{\varepsilon^{-2}s}) + 1 - 2 \tr(\rho_{\varepsilon^{-2}s} P_{\hat{\mathcal {N}}_{\varepsilon^{-2}s-T}})\\
						&\leq  2 \tr(\rho_{\varepsilon^{-2}s} (1 - P_{\hat{\mathcal {N}}_{\varepsilon^{-2}s-T}})),
\end{align*}
where we have used that $\tr(\rho_{\varepsilon^{-2}s}) = 1$ and $\tr(\rho^2_{\varepsilon^{-2}s}) \leq 1$. We use definition \eqref{E:rhos} of the posterior state $\rho_{s}$ to show that
\begin{align*}
\int  \tr\Big (\rho_{\varepsilon^{-2}s} (1 - P_{\hat{\mathcal {N}}_{\varepsilon^{-2}s-T}})\Big ) d \mathbb{P}^{\varepsilon}_{\rho}  &  = \int \tr \Big ( {\sigma}_{\varepsilon}^{(\varepsilon^{-2} s,0)}[\rho]  (1 - P_{\hat{\mathcal {N}}_{\varepsilon^{-2}s-T}})\Big )  d \mathbb{P} \\ & = \sum_{\nu}  \int_{ (E_{\nu}^{(\varepsilon^{-2} s, \varepsilon^{-2} s - T)}) }  \tr \Big ( {\sigma}_{\varepsilon}^{(\varepsilon^{-2} s,0)}[\rho]  (1 - P_{\nu})\Big )  d \mathbb{P} .
\end{align*} 
Using  the factorization property \eqref{factorization}, and the Cauchy-Schwarz inequality, we find that 
\begin{align}
\label{eq:hawai}
{\mathbb{E}^{\varepsilon}_\rho}  \Big [ \|\rho_{\varepsilon^{-2}s}-P_{\hat{\mathcal {N}}_{\varepsilon^{-2}s-T}}\|_2 \Big ]^2  \leq  2  \sum_{\nu \in \sigma(\mathcal{N})} & \tr({\sigma}_\varepsilon^{(\varepsilon^{-2} s, \varepsilon^{-2} s - T)}(E_{\nu}^{(\varepsilon^{-2} s, \varepsilon^{-2} s - T)}) \notag \\  &  \times  { \Big (  \mathbb{E}\Big [{\sigma}_{\varepsilon}^{(\varepsilon^{-2} s- T,0)}[\rho]\Big ]  \Big )  }  (1-P_\nu)).
\end{align}
By Lemma~\ref{lem:5.1}, and Lemma \ref{prop:unraveling}, $\mathcal{P}_\perp \rho_i$ is of order $\varepsilon$  (here we use that $s > 2 \varepsilon^2 T$), hence 
$ \Big{(}  {\sigma}_\varepsilon^{(\varepsilon^{-2} s, \varepsilon^{-2} s - T)}(E_{\nu}^{(\varepsilon^{-2} s, \varepsilon^{-2} s - T)}) 
 \mathcal{P}_\perp   \Big{)}    \Big{(}  \mathcal{P}_\perp    { \Big (  \mathbb{E}\Big [{\sigma}_{\varepsilon}^{(\varepsilon^{-2} s- T,0)}[\rho]\Big ]  \Big )  }     \Big{)}   $ is of order $\varepsilon^2$; 
 (here we use  Lemmas~\ref{lemND} and \ref{lem:5.4}, with $\varepsilon$ sufficiently small). Using Lemmas~\ref{lemND} and \ref{lem:5.4} again, we prove that the right side of Eq.~(\ref{eq:hawai}) is of order 
 $\varepsilon^2 T  $, provided
 $\varepsilon$ is small enough, depending on $\alpha$. (Notice that $\tr ( \mathcal{P}_\perp \tilde \rho = 0 ) $, for every $\tilde \rho$).
 The term $\varepsilon^2 T$ is dominated by $\varepsilon$ for the choice $T = \alpha \log \varepsilon$.
 \end{proof}
\newpage

\appendix
\section{Underlying measure spaces}
\label{App:0}
The probability measure that we define in Section \ref{ET} has the following form 
 \begin{equation}\label{ESta}
\mathbb{P} \equiv P \otimes \mu^{ \otimes  \mathbb{N}},
\end{equation}
where $P:=  \upsilon ^{ \otimes \mathbb{N}} $ is the standard Poisson measure defined on the set 
$ \Omega:=   [0, \infty)^\mathbb{N}   $ (equipped with the sigma algebra generated by cylinder sets), with 
$ \upsilon  $ the measure on $   [0, \infty)  $ with density $e^{-r }$, $r \in [0, \infty)$.

 The statistics of measurement results $\xi_{N_{u}+ 1}, \dots \xi_{N_s}$ in a time interval $(u,s)$ is independent of the measurement results outside of this interval.  We next describe a mathematical construction that captures this fact:
we define the set $  \Xi_0  $ to be the disjoint union $   \Xi_0: =  \bigcup_{n \in \mathbb{N}}\Big( [0, \infty)^{n} \times \mathcal{X}^{n-1} \Big )    $, where $\mathcal{X}^0 = \emptyset $. We endow $ \Xi_0 $ with the sigma algebra $ {\mathcal{F}}_0$ whose elements have the form  $ A = \bigcup_{n \in \mathbb{N}} A_{n} $, with $ A_n  $ in the - product - sigma algebra associated to $     [0, \infty)^{n} \times \mathcal{X}^{n-1}   $. For arbitrary $ s > u \geq 0 $, we define a map $ \theta^{(s, u)} :  \Xi  \to \Xi_0  $ by setting 
\begin{equation}\label{great fun}
\theta^{(s, u)}(\underline{\tau}, \underline{\xi}) : = ( t_{N_{u} + 1} - u, \tau_{N_u + 2}, \cdots, \tau_{N_s}, s - t_{N_s};  \xi_{N_u+1}, \cdots, \xi_{N_s}     )
\end{equation}
if $N_s >N_u$, and
$$
\theta^{(s, u)}(\underline{\tau}, \underline{\xi}) : = (s-u)  \quad \mbox{if} \quad N_s = N_u.
$$
This definition only makes sense if $N_s - N_u < \infty$; but luckily the set

\begin{equation}\label{MA}
\mathcal{A} := \Big \{ (\underline{\tau}, \underline{\xi}) \in \Xi  \Big |    N_s(\underline{\tau}) < \infty, \text{for every} \,  s \geq 0    \Big \} 
\end{equation}
has full measure.

 Basic properties of the Poisson process imply that 
$ \theta^{(s, u)} $ and $ \theta^{(s- u, 0)}  $ have the same distribution and that, for every m-tuple of times $  0 \leq  s_1 < s_2 < \cdots < s_m   $, the  the random elements 
$ \theta^{(s_m, s_{m-1})}, \theta^{(s_{m-1}, s_{m-2})}, \cdots   ,\theta^{(s_{2}, s_{1})} $ are independent.  This, in turn, implies that, for arbitrary complex-valued, integrable functions $f_1,  \cdots, f_m   : \Xi_0 \to \mathbb{C} :$
\begin{align}\label{Factor}
\mathbb{E}\big ( f_m \circ \theta^{(s_m, s_{m-1})}  \cdots f_1 \circ \theta^{(s_2, s_{1})}  \big ) = \mathbb{E}\big ( f_m \circ \theta^{(s_m, s_{m-1})}  \big )  \cdots
\mathbb{E}\big ( f_1 \circ \theta^{(s_2, s_{1})}  \big )
\end{align}  
and this also holds if $ f_1, \cdots, f_m $ are matrix-valued.

We define  $ \Theta :   \Xi_0 \to \mathcal{B}(\mathcal{B}(\mathcal{H})) $ as follows: 
\begin{align} \label{Thetat}
\Theta (s_1, s_2, \cdots, s_{n+1};  \xi_1, \cdots, \xi_{n}) =   e^{-i \varepsilon s_{n+1} \adj_H}  \Phi_{\xi_{n}}   \dots   e^{-i \varepsilon s_2 \adj_H} \Phi_{\xi_{1}} e^{-i \varepsilon s_1 \adj_H},   
\end{align} 
and $ \Theta(s) =    e^{-i \varepsilon s \adj_H } $.
The time evolution of the system in Eq.   \eqref{prop} has the following expression:   
\begin{align} \label{misto}
{\sigma}_\varepsilon^{(s,u)} := \Theta \circ \theta^{(s, u)}.
\end{align}
 Note that the fact that 
 $\theta^{(s,u)} $   and  $  \theta^{(s-u,0)} $  have the same distribution implies that ${\sigma}_\varepsilon^{(s,u)}$ and ${\sigma}_\varepsilon^{(s-u,0)}$ have the same distribution. 
For arbitrary sets $ A, B \in {\mathcal{F}}_0 $ and $0 \leq u < s$,   Eqs.~\eqref{Factor} and (\ref{pgp}) imply that
\begin{equation}
\label{factorizationB}
 {\sigma}^{(s,0)}_{\varepsilon}\Big ( \big (  \theta^{(s,u)} \big )^{-1}(A) \cap 
   \big ( \theta^{(u,0)} \big )^{-1}(B) \Big ) = {\sigma}_\varepsilon^{(s, u)}\Big ( \big (  \theta^{(s,u)} \big )^{-1}(A)  \Big ) {\sigma}_\varepsilon^{(u, 0)} \Big (  \big(  \theta^{(u,0)} \big )^{-1}(B)   \Big ).
 \end{equation}

\section{Estimates Required in the Proof of Lemma \ref{lem:5.4}}

\label{App:A}

In this section we use two additional super-operator norms  
\begin{equation}\label{nomames}
\| \mathcal{X} \|_{1, \text{op}} := \underset{\|  X  \|_{1} =1 }{\sup} \| \mathcal{X} X \|_{1}, \quad  \| \mathcal{X} \|_{\infty, \text{op}} := \underset{\|  X  \| =1 }{\sup} \| \mathcal{X} X \|,
\end{equation}
where $\|X\|_1$ is the trace norm on $B(\mathcal{H})$.

\begin{lemma}
\label{lem:CPDS}
Let $\mathcal{K}$ be a completely positive map on $B(\mathcal{H})$ for which 
$$\mathcal{K} \id \leq \id, \quad \mbox{and} \quad  \mathcal{K}^* \id \leq \id,$$ 
then $$\|\mathcal{K}\|_{1,op} \leq 1,\,  \|\mathcal{K}\|_{2,op} \leq 1, \, \mbox{and} \,\, \|\mathcal{K}\|_{\infty, op} \leq 1.$$
\end{lemma}

\begin{proof} 
The claim about the infinity norm is standard \cite[Cor 3.2.6]{BR}, however it is typically stated in the most important case $\mathcal{K} \id = \id$ so we give a full proof here for readers convenience.
By Kadison's inequality 
$$
\|\mathcal{K}(A)\|^2 = \|\mathcal{K}(A)^* \mathcal{K}(A)\| \leq \|\mathcal{K}(A^* A)\| \| \mathcal{K}(\id)\|,
$$
holds for any operator $A$. Since $\mathcal{K}( A^* A) \leq \|A^*A\| \mathcal{K}(\id)$, we conclude that
$$
\|\mathcal{K}(A)\|^2 \leq \| \mathcal{K}(\id) \|^2 \|A^*A \| \leq \|A \|^2.
$$
The statement about the trace norm follows by a duality argument, $\|\mathcal{K}\|_{1,op} = \|\mathcal{K}^* \|_{\infty,op}$.

We now prove the remaining claim. 
By the general theory of completely positive maps \cite{Holevo}, there exists {a} finite index set $\mathcal{I}$ and  operators $\Gamma_\alpha, \alpha \in \mathcal{I}$ such that 
$$ 
\mathcal{K} X = \sum_{\alpha \in \mathcal{I}} \Gamma_\alpha X \Gamma_\alpha^*.
$$  
By Cauchy-Schwartz inequality we then have  
$$ 
[\tr(X^* \mathcal{K} Y)]^2 \leq \left(\sum_\alpha \tr(\Gamma_\alpha^* X X^* \Gamma_\alpha)\right) \left(\sum_\alpha \tr(\Gamma_\alpha Y^* Y \Gamma^*_\alpha)\right).
$$  
Using the cyclicity of trace and assumptions  $\mathcal{K} \id \leq \id$ and $\mathcal{K}^* \id \leq \id$, we conclude $\tr(X^* \mathcal{K} Y)^2 \leq \|X\|_2^2 \|Y\|_2^2$, which implies $\|\mathcal{K}\|_{2,op} \leq 1$ (we use that for   self-adjoint positive bounded operators   $B $ and  $ D$, $B \leq D$ implies that $\| B\| \leq \| D \|   $; moreover, if $D$ is trace class $|  \tr (BD) | \leq    \| B D  \|_1
\leq \| B \| \| D \|_1 $).
\end{proof} 
The previous lemma applies in particular to operators $e^{i \varepsilon \ad_H}$ and $\mathcal{P}$. We use this in the proof of the following lemma.

\begin{lemma}
\label{lem:delta}
We have that 
$$
\| e^{i \varepsilon \adj_H}  \|_{2, \text{op}}  = 1, \qquad  \| e^{i \varepsilon \adj_H} -1 \|_{2, \text{op}} \leq 2 \varepsilon \|H\|, \quad \| \mathcal{P} (e^{i \varepsilon \adj_H} -1)  \mathcal{P}\|_{2, \text{op}} \leq 4 \varepsilon^2 \|H \|^2.
$$
\end{lemma}
\begin{proof} 
The operator $\adj_H$ is a bounded self-adjoint operator on the Hilbert space $\mathcal{J}_2(\mathcal{H})$, then the first inequality follows directly from functional calculus.  The proof of the second and third inequalities uses two elementary trigonometric inequalities,
$$
|e^{ix} - 1| \leq |x|, \quad |e^{ix} - 1 -ix| \leq x^2,
$$
valid for all $x \in \mathbb{R}$.
. The functional calculus then gives,
$$
\| e^{i \varepsilon \adj_H} -1 \|_{2,op} \leq   \varepsilon \|\adj_H \|_{2,op},
$$
and the first inequality is  established in view of $\|\adj_H\|_{2,op} \leq 2 \|H \|$. The second inequality follows in a similar manner using an inequality
$$
\| e^{i \varepsilon \adj_H} -1 - i \varepsilon \adj_H \|_{2,op} \leq  \varepsilon^2 \|\adj_H \|^2_{2,op},
$$
in conjunction with $\mathcal{P} \adj_H \mathcal{P} =0$ and $\|\mathcal{P}\|_{op,2} =1$.
\end{proof}

The following lemma lists bounds on the terms appearing in the expansion of ${\sigma}_\varepsilon^{(s,u)} - {\sigma}_0^{(s,u)}$ up to the third order in $\varepsilon$. These terms are used and defined in the proof of Lemma~\ref{lem:5.4}.

\begin{lemma} 
\label{lem:a2}
Fix $ \underline{\tau} \in \Omega$ and let $E$ be a measurable set
 on $\mathcal{X}^{\mathbb{N}}$, then 
\begin{align}
\Big{\|} \int_E \mathcal{P} A^{(1)}_j(\underline{ \tau }, \underline{\xi})\mathcal{P} \d \mu^{\otimes \mathbb{N}}(\underline{\xi}) \Big{\|}_{2,op} &\leq 4 d^2 \varepsilon^2 (  \tau_{j+1})^2  \|H \|^2\label{eq:a2_1}\\
 \Big{\|}  \int_E \mathcal{P} A^{(2)}_{j,k}(\underline{ \tau }, \underline{\xi})\mathcal{P} \d \mu^{\otimes \mathbb{N}}(\underline{\xi})  \Big{\|}_{2,op} &\leq 16 d^3 \varepsilon^4 \|H \|^4 (  \tau_{j+1} )^2 (  \tau_{k+1} )^2 \nonumber\\
& + 4 d^4 \varepsilon^2 \|H\|^2   \tau_{j+1}   \tau_{k+1}  (1-g)^{j-k-1}  \label{eq:a2_2}\\
\Big{\|}   \int_E  R^{(3)}_{j,k,l}(\underline{  \tau  }, \underline{\xi}) \d \mu^{\otimes \mathbb{N}}(\underline{\xi})  \Big{\|} _{2,op} &\leq 8 \varepsilon^3  \tau_{j+1}    \tau_{k+1}   \tau_{l+1}   \|H \|^3 \label{eq:a2_3}\\
 \Big{\|}  \int_E \mathcal{P}_\perp A^{(1)}_j(\underline{ \tau  }, \underline{\xi}) \d \mu^{\otimes \mathbb{N}}(\underline{\xi})  \Big{\|} _{2,op} &\leq 2 d^4 \varepsilon
     \tau_{j+1} \|H \| (1-g)^{N_s -j-1} \label{eq:a2_4}\\
\Big{\|} \int_E  R^{(2)}_{j,k}(\underline{  \tau  }, \underline{\xi}) \d \mu^{\otimes \mathbb{N}}(\underline{\xi})  \Big{\|} _{2,op} &\leq 4 \varepsilon^2   \tau_{j+1}   \tau_{k+1}  \|H \|^2 \label{eq:a2_5},
\end{align}
where $d$ denotes the dimension of $\mathcal{H}$.
\end{lemma}
\begin{proof} The boundary cases for which any index takes value $N_s$ or $N_u$ need to be treated separately.  For a generic index $j$, the perturbation 
$B_j = e^{-i \varepsilon \tau_{j+1} \ad_H }  $ $ 
  = e^{-i \varepsilon ( t_{j+1} - t_j) \ad_H }    $, but for $j=N_s$ the time $t_{j+1}$ is replaced by $s$ and for $j=N_u$ the time $t_j$ is replaced by $u$. For the sake of clarity we exclude these cases from the proof (they can be analyzed with the method me present below). 

 We start with a proof Eq.~(\ref{eq:a2_3}). We claim that for any matrices $X, Y$ and fixed $j,k,l$ there holds the inequality 
$$
\Big{|}  \tr\Big (X^* \int_E  R^{(3)}_{j,k,l}(\underline{ {\tau} }, \underline{\xi}) \d \mu^{\otimes \mathbb{N}}(\underline{\xi}) Y \Big ) \Big{|}   \leq 8 \varepsilon^3
  \tau_{j+1} \tau_{k+1}  \tau_{l+1}  \|X \|_2 \|Y\|_2 \|H \|^3,
$$
which is equivalent to the bound Eq.~(\ref{eq:a2_3}). To prove the inequality we use the identity $(B_j - \mathsf{1})X = (e^{-i \varepsilon \tau_{j+1}  H} -1) X + e^{-i \varepsilon \tau_{j+1} H} X (e^{i \varepsilon \tau_{j+1}  H} -1)$ and expand the expression under the trace to corresponding eight terms. Denoting  $S_j(\cdot) = (e^{ -i \varepsilon \tau_{j+1}  H} -1) \cdot (e^{i \varepsilon \tau_{j+1}  H} -1)$ we obtain by the Cauchy-Schwartz inequality
\begin{align*}
\Big | \int_E \tr(X^* R^{(3)}_{j,k,l} Y) \Big | \leq 
 \sum_{Z_j,Z_k,Z_l} & \left[ \int_{E} \tr( X X^* {\sigma}_0^{(s, t_j)} Z_j {\sigma}_0^{(t_j, t_k)} Z_k {\sigma}_0^{(t_k, t_l)} Z_l {\sigma}_\varepsilon^{(t_l, u)} \id) \right. \\
& \times \left. \int_{E}  \tr( {\sigma}_0^{(s, t_j)} Z^\#_j {\sigma}_0^{(t_j, t_k)} Z^\#_k {\sigma}_0^{(t_k, t_l)} Z^\#_l {\sigma}_\varepsilon^{(t_l, u)} (Y^* Y)) \right]^{\frac{1}{2}},
\end{align*}
where the sum consists of eight terms in which pairs $(Z_\alpha, Z_\alpha^\#)$ take either value $( B_\alpha, S_\alpha)$ or $(S_\alpha,1)$, $\alpha=j,\,k,\,l$. 

Since all traces are now positive, we are in position to extend the integrals to all $\mathcal{X}^{\mathbb{N}} $ to obtain 
\begin{align*}
\Big | \int_E \tr(X^* R^{(3)}_{j,k,l} Y) \Big | \leq 
 &\sum_{Z_j,Z_l,Z_l}  \left[ \tr( X X^* {\sigma}_0^{(s, t_j)}
 (   \mathcal{X}^{\mathbb{N}}   ) Z_j {\sigma}_0^{(t_j, t_k)} Z_k {\sigma}_0^{(t_k, t_l)}(  \mathcal{X}^{\mathbb{N}}  ) Z_l {\sigma}_\varepsilon^{(t_l, u)} (  \mathcal{X}^{\mathbb{N}}   ) \id) \right. \\
& \times \left. \tr( {\sigma}_0^{(s, t_j)}(   \mathcal{X}^{\mathbb{N}}  ) Z^\#_j {\sigma}_0^{(t_j, t_k)}
(  \mathcal{X}^{\mathbb{N}}   ) Z^\#_k {\sigma}_0^{(t_k, t_l)}(   \mathcal{X}^{\mathbb{N}}  ) Z^\#_l {\sigma}_\varepsilon^{(t_l, u)}(   \mathcal{X}^{\mathbb{N}}  ) (Y^* Y)) \right]^{\frac{1}{2}},
\end{align*}
where we denoted 
$$
{\sigma}_\varepsilon^{(s,t)}(   \mathcal{X}^{\mathbb{N}}  ) = \int_{       \mathcal{X}^{\mathbb{N}}     } {\sigma}_\varepsilon^{(s,t)}(\underline{\tau}, \underline{\xi}) \d \mu^{\otimes \mathbb{N}}(\underline{\xi}).
$$
By its definition, ${\sigma}_\varepsilon^{(s,t)}(  \mathcal{X}^{\mathbb{N}}   )$ is doubly-stochastic completely positive map, and hence $\|{\sigma}_\varepsilon^{(s,t)}(  \mathcal{X}^{\mathbb{N}}   )\|_{1,op} = \|{\sigma}_\varepsilon^{(s,t)}(   \mathcal{X}^{\mathbb{N}}   )\|_{\infty,op} =1$ (see Lemma \ref{lem:CPDS}). Also $\|B_j\|_{1, op} = \|B_j\|_{\infty, op} =1$, and the operator bound $\|e^{i \varepsilon H} -\id\| \leq \varepsilon \|H\|$ implies $\|S_j\|_{1, op} \leq \varepsilon^2
 (  \tau_{j+1} )^2 \|H\|^2$ as well as $\|S_j\|_{\infty, op} \leq \varepsilon^2 (\tau_{j+1} )^2 \|H\|^2$. Noting that in each eight terms in the sum there appears exactly one $S_\alpha$ for each $\alpha = j,k,l$ we conclude, from the inequalities stated above, that each of the eight terms  is bounded by $\varepsilon^3   \tau_{j+1}    \tau_{k+1}   \tau_{l+1}  \|X\|_2 \|Y\|_2 \|H\|^3$ (we use that for any trace class operator $A$ and every bounded operator $B$: $ | \tr(AB) | \leq 
\| AB \|_1  \leq \| B \| \|  A\|_1 $). 

The proof of Eq.~(\ref{eq:a2_5}) is analogous. The only difference is that there are only two perturbations. This eliminates the product $ \tau_{l+1}  \varepsilon \|H\|$ and a combinatorial factor of $2$.

We now proceed with the proof of Eqs.~(\ref{eq:a2_1}, \ref{eq:a2_4}). We express the super-operator $A_j^{(1)}$ in the basis consisting of operators $\ket{\nu}\bra{\nu'}$ formed from the eigenvectors of $\mathcal{N}$. For eigenvalues $\nu_1, \dots, \nu_4$ we have
\begin{multline*}
\tr(  \ket{\nu_1}\bra{\nu_2} A_j^{(1)} \ket{\nu_3} \bra{\nu_4}) = V_{\xi_{N_s}}(\nu_1) \overline{V}_{\xi_{N_s}}(\nu_2) \dots V_{\xi_{j+1}}(\nu_1) \overline{V}_{\xi_{j+1}}(\nu_2) \\ \times \tr(  \ket{\nu_1}\bra{\nu_2} (B_j - \id)\ket{\nu_3} \bra{\nu_4}) V_{\xi_{j}}(\nu_4) \overline{V}_{\xi_{j}}(\nu_3) \dots V_{\xi_{N_u+1}}(\nu_4) \overline{V}_{\xi_{N_u+1}}(\nu_3).
\end{multline*}
Hence for any matrices $X,\, Y$ we have
\begin{multline}
\label{eq:214}
\Big |\tr\Big ( X^* \int_E A_j^{(1)} Y  \Big )\Big |  \leq \sum_{\nu_1, \dots, \nu_4} |  \overline{ X_{\nu_2 \nu_1}} Y_{\nu_3 \nu_4} \tr(\ket{\nu_1}\bra{\nu_2}(B_j - \id) \ket{\nu_3} \bra{\nu_4})| \\
\times 
\int_{\Xi} |V_{\xi_{N_s}}(\nu_1)  \overline{V}_{\xi_{N_s}}(\nu_2) \dots  V_{\xi_{j+1}}(\nu_1)  \overline{V}_{\xi_{j+1}}(\nu_2) V_{\xi_{j}}(\nu_4)  \dots  \overline{V}_{\xi_{N_u+1}}(\nu_3)| \d \mu^{\otimes \mathbb{N}}(\underline \xi),
\end{multline}
where $X_{\nu \nu'} = \bra{\nu'} X \ket{\nu}$. Note that we extended the integration region after estimating the integrand by its absolute value.

For $X, Y$ diagonal in the eigenbasis basis of $\mathcal{N}$ we have $\nu_1 = \nu_2$ and $\nu_3 = \nu_4$ in Eq.~(\ref{eq:214}). The second line in that equation can be then estimated by $1$ and hence we get
$$
\tr\Big ( X^* \int_E\mathcal{P} A_j^{(1)} \mathcal{P} Y \Big ) \leq \sum_{\nu, \nu' \in \sigma(\mathcal{N})} |X_{\nu \nu}| |Y_{\nu' \nu'}| |\tr(\ket{\nu'} \bra{\nu'}  \mathcal{P}  (B_j - \id)  \mathcal{P}  \ket{\nu} \bra{\nu})|.
$$
 Using Cauchy-Schwartz inequality we then conclude
$$
\tr \Big ( X^* \int_E\mathcal{P} A_j^{(1)} \mathcal{P} Y  \Big ) \leq \|X\|_2 \|Y\|_2 \dim(\mathcal{H})^2 \|\mathcal{P}(B_j - \id)\mathcal{P}\|_{2,op}
$$
and Eq~(\ref{eq:a2_1}) is established using the bound in Lemma~\ref{lem:delta}.

On the other hand, if $X$ is off-diagonal, i.e. $\bra{\nu} X \ket{\nu} = 0$ for all $\nu \in \sigma(\mathcal{N})$, then $\nu_1 \neq \nu_2$ and the integrand on the second line of Eq.~(\ref{eq:214}) is bounded by $(1-g)^{N_s - j -1}$ in view of Eq.~(\ref{Eq:g_def}). This gives
$$
\tr \Big ( X^* \int_E\mathcal{P}_\perp A_j^{(1)} Y \Big ) \leq (1-g)^{N_s -j -1} \sum_{\nu_1 \dots \nu_4 } |\overline{ X_{\nu_2 \nu_1}}|  |Y_{\nu_3 \nu_4}| |\tr(\ket{\nu_1} \bra{\nu_2} (B_j - \id) \ket{\nu_3} \bra{\nu_4})|.
$$
 Using Cauchy-Schwartz inequality we then conclude
$$
\tr \Big ( X^* \int_E\mathcal{P} A_j^{(1)} \mathcal{P} Y \Big ) \leq (1-g)^{N_s -j -1} \|X\|_2 \|Y\|_2 \dim(\mathcal{H})^4 \|(B_j - \id)\|_{2,op}
$$
and Eq~(\ref{eq:a2_4}) is established using the bound in Lemma~\ref{lem:delta}.

It remains to prove Eq.~(\ref{eq:a2_2}). We insert a decomposition of identity, $\id = \mathcal{P} + \mathcal{P}_\perp$, to express the LHS of the equation as 
\begin{align*}
  \mathcal{P}A_{j,k}^{(2)} \mathcal{P} = &\mathcal{P} {\sigma}_{0}^{(s,t_{j})}  (B_{j} - \mathds{1}) \mathcal{P} {\sigma}_{0}^{(t_{j},t_{k})}   \mathcal{P} (B_{k} - \mathds{1}) {\sigma}_{0}^{(t_{k},u)} \mathcal{P} \\
  &+\mathcal{P} {\sigma}_{0}^{(s,t_{j})}  (B_{j} - \mathds{1}) \mathcal{P_\perp} {\sigma}_{0}^{(t_{j},t_{k})}   \mathcal{P_\perp} (B_{k} - \mathds{1}) {\sigma}_{0}^{(t_{k},u)} \mathcal{P}.
\end{align*}
Proceeding as in the proof of Eq.~(\ref{eq:a2_1}) we get an estimate for the integral of the first line,
\begin{multline*}
\Big \|\int_{E} \mathcal{P} {\sigma}_{0}^{(s,t_{j})}  (B_{j} - \mathds{1}) \mathcal{P} {\sigma}_{0}^{(t_{j},t_{k})}   \mathcal{P} (B_{k} - \mathds{1}) {\sigma}_{0}^{(t_{k},u)} \mathcal{P} \Big \|_{2,op} \\ \leq 16 \varepsilon^4 \|H \|^4 \dim(\mathcal{H})^3 (\tau_{j+1})^2 (\tau_{k+1} )^2.
\end{multline*}
Proceeding as in the proof of Eq.~(\ref{eq:a2_4}) we get an estimate for the integral of the second line,
\begin{multline*}
\Big \|\int_{E} \mathcal{P} {\sigma}_{0}^{(s,t_{j})}  (B_{j} - \mathds{1}) \mathcal{P}_\perp {\sigma}_{0}^{(t_{j},t_{k})}   \mathcal{P}_\perp (B_{k} - \mathds{1}) {\sigma}_{0}^{(t_{k},u)} \mathcal{P} \Big \|_{2,op} \\ \leq 4 \varepsilon^2 \|H \|^2 \dim(\mathcal{H})^4 (1-g)^{j - k -1}   \tau_{j+1}  \tau_{k+1} .
\end{multline*}
This finishes the proof.
\end{proof}

\section{Bounds on expectation values over the Poisson process} \label{B1} 
The proof of Lemma~\ref{lem:5.4} required estimating expectation values of certain random variables over the Poisson process. Here we justify these estimates. Stopping time theory for sums of i.i.d. random variables is a convenient level of generality to present the bounds.

Let $(X_i)_{i=1}^{\infty}$ be positive i.i.d. random variables and $T$ a stopping time with respect to them. We assume that for any $k$ positive $\mathbb{E}[X_1^k] \leq 1$ and that all moments of $T$  are finite. By Walds equation (we use $\boldsymbol{E}$ to denote the corresponding expected value),
\begin{equation}
\label{eq:a1.p1}
\boldsymbol{E}[\sum_{i=1}^T X_i] = \boldsymbol{E}[X_1] \boldsymbol{E}[T] \leq \boldsymbol{E}[T].
\end{equation}
Now consider a bounded positive function $f : \mathbb{N} \to \mathbb{R}$. By optional stopping time theorem (for $T < 2 $ the sum below is set to be zero),
$$
\boldsymbol{E}[\sum_{i=1}^{T- 1} \sum_{j=i+1}^T f(j-i) X_i (X_j - \boldsymbol{E}[X_j])] = 0,
$$
and using the assumption that the mean value of $X_j$'s is bounded by $1$, we get
\begin{align*}
\boldsymbol{E}[\sum_{i=1}^{T-1} \sum_{j=i+1}^T f(j-i) X_i X_j] &\leq \boldsymbol{E}[\sum_{i=1}^{T-1} \sum_{j=i+1}^T f(j-i) X_i]\\
										&\leq \left(\boldsymbol{E}[\sum_{i=1}^{T-1} X_i^2]\boldsymbol{E}[\sum_{i=1}^{T-1} \left(\sum_{j=i}^T f(j-i)\right)^2 ]	 \right)^\frac{1}{2}.			
\end{align*}
Using Eq.~(\ref{eq:a1.p1}) with $X_i$ replaced by $X_i^2$, we then conclude
\begin{equation*}
\boldsymbol{E}[\sum_{i=1}^{T-1} \sum_{j=i+1}^T f(j-i) X_i X_j]  \leq (\boldsymbol{E}[T])^{\frac{1}{2}}\left(\boldsymbol{E}[\sum_{i=1}^{T-1} \left(\sum_{j=i}^T f(j-i)\right)^2 ]\right)^\frac{1}{2}.
\end{equation*}
This equation is used either for $f=1$ in which case it gives
\begin{equation}
\label{eq:a1.p2}
\boldsymbol{E}[\sum_{i=1}^{T-1} \sum_{j=i+1}^T X_i X_j]  \leq (\boldsymbol{E}[T])^{\frac{1}{2}} (\boldsymbol{E}[T^3])^\frac{1}{2},
\end{equation}
or for a summable function $f$ for which it gives
\begin{equation}
\label{eq:a1.p3}
\boldsymbol{E}[\sum_{i=1}^{T-1} \sum_{j=i+1}^T f(j-i) X_i X_j] \leq \boldsymbol{E}[T] \sum_{k=1}^\infty f(k).
\end{equation}
Finally, to bound sums of three point correlation functions we again use optional stopping time theorem and Cauchy-Schwartz inequality (for $T < 3 $ the sum below is set to be zero),
\begin{align*}
\boldsymbol{E}[\sum_{j=1}^{T-2} \sum_{k=j+1}^{T-1} \sum_{l= k+1}^T X_j X_k X_l] &\leq \boldsymbol{E}[\sum_{j=1}^{T-2} \sum_{k=j+1}^{T-1} \sum_{l= k+1}^T X_j X_k] \\
														&\leq \left([\boldsymbol{E}[\sum_{j=1}^{T-2} \sum_{k=j+1}^{T-1}  X_j^2 X_k^2]\right)^{\frac{1}{2}} \left([\boldsymbol{E}[\sum_{j=1}^{T-2} \sum_{k=j+1}^{T-1} (T -k)^2] \right)^\frac{1}{2}.
\end{align*}
Using Eq.~(\ref{eq:a1.p2}) with $X_j$ replaced by $X_j^2$ we conclude,
\begin{equation}
\label{eq:a1.p4}
\boldsymbol{E}[\sum_{j=1}^{T-2} \sum_{k=j+1}^{T-1} \sum_{l= k+1}^T X_j X_k X_l]  \leq (\boldsymbol{E}[T])^{\frac{1}{4}} (\boldsymbol{E}[T^3])^\frac{1}{4} (\boldsymbol{E}[T^4])^\frac{1}{2}.
\end{equation}
We note that in the proof of Lemma~\ref{lem:5.4} these inequalities are used in a setting where for all integers $k$, $\boldsymbol{E}[T^k] \leq C \boldsymbol{E}[T]^k$ holds for some constant $C$.

\section{Notation Index}
\label{App:D} 

In this section we compile all relevant notation we use in this paper. However, we do not report notation that is only locally used, i.e. notation that is only valid within the limits of a certain proof.         
\vspace{.5cm}

\begin{tabular}{ p{7cm} l     }
 $  \mathcal{H}, \, \mathcal{B}(\mathcal{H}), \ad_H $ & Above and Below Eq. \eqref{nomame} \\   
 $  \|  \cdot  \|  $, $  \|  \cdot  \|_2  $,  $  \|  \cdot  \|_1  $, $  \|  \cdot  \|_{2, {\rm op}}  $, $  \|  \cdot  \|_{1, {\rm op}}  $, $  \|  \cdot  \|_{\infty, {\rm op}}  $    & Eqs. \eqref{nomame} and \eqref{nomames} \\
 $  \mathcal{N}$, $P_\nu$  &  Below Eq. \eqref{normalis} and Eq. \eqref{eq:spectral_decomposition_of_C} \\
$  (\mathcal{X}, \sigma)$  & Above Eq. \eqref{normalis} \\
$ \mu, \, V_{\xi}(\nu)  $, $V_\xi$  & Above and below Eq. \eqref{normalis} \\ 
$\Phi_\xi $ & Below Eq. \eqref{eq:spectral_decomposition_of_C} \\  
$ (\Xi, {\mathcal{F}}),  \mathbb{P}, \mathbb{E} , \,  \underline{\tau}, \, \underline{\xi}, \,  t_j $ 
  & Beginning of Section \ref{ET} and Appendix \ref{App:0}  \\
 $N_s$, $ H $, $\varepsilon  $ & Above Eq. \eqref{prop} \\
$ \sigma_{\varepsilon}^{(s,u)} $ &  Eq. \eqref{prop} \\ 
$ \sigma_{\varepsilon}^{(s,u)} (E)$ &  Eq. \eqref{para} \\
$ \mathbb{P}^\varepsilon_{\rho}(E),  \,  \mathbb{E}^\varepsilon_{\rho}, \,   \mathbb{P}^0_{\rho}(E) $ &  Eq. \eqref{eq:2p} and below, Eq. \eqref{eq:5} \\  
 $ \rho_s (\underline{\tau}, \underline{\xi} )$ &  Eq. \eqref{E:rhos} \\ 
 $ f(\xi | \nu)  $ &  Eq. \eqref{pto} \\  
  $  \mu_\nu $ &  Eq. \eqref{E:munu} \\
$ l_s^T (\nu | \underline{\xi}) $ & Eq. \eqref{eq:likelihood} \\
$ \n_s(T) $ & Eq. \eqref{estin}\\
$ \mathcal{M}_t$ & Eq. \eqref{proc} and below \\
$ Y_s $ and $ Q	(\nu', \nu) $ & Eq. \eqref{Qdef} and above \\
$  \pi_{\rho}(\nu), \Gamma_t $  &  Above Eq. \eqref{Qdef} \\
$  T(\varepsilon )  $ & Theorem \ref{thm:case1} \\
  $\mathcal{L}_{\varepsilon}$, $ \Phi $  & Eq. \eqref{eq:Lin}   \\
   \end{tabular}

\begin{tabular}{ p{7cm} l     } 
$ \mathcal{P} $, $ \mathcal{P}_{\nu} $   &   Eqs. \eqref{pndjo},  \eqref{pen} \\ 
$  \mathcal{P}_{\perp}, \,  \langle \cdot, \cdot \rangle   $ &  Eq. \eqref{IN} and Above \\
$g_{sp}$  & Below Eq. \eqref{eq:decoupling}\\
$ g  $ &  Eq. \eqref{Eq:g_def} \\
$  \mathcal{L}_{\varepsilon}^{\perp}$
   & Above Eq. \eqref{tointegrate}  \\
$ 
 \, Q   $ & Eq.  \eqref{Qs} \\
$  A_{j}^{(1)}, \,  A_{j, k}^{(2)}, \, R_{j, k}^{(2)}, \, R_{j, k, l}^{(3)}, \, B_j 
  $ & Lemma \ref{lem:5.4} \\ 
$  E^{(T+s, s)}_{\nu}, \, \mathbf{E}_{\mu_{\nu}} $ &  Eq.  \eqref{Enu} and below \\ 
$  X_{\nu, \nu'}(\xi), \, I $ & Eq. \eqref{def_I} and above.  \\
$P, \, \upsilon, \, \Omega, \, \Xi_0, \mathcal{F}_0$ & Beginning of Appendix \ref{App:0} \\ 
$ \theta^{(s, u)}(\underline{\tau}, \underline{\xi}) $ & Eq. \eqref{great fun} and below \\ 
$\mathcal{A} $ & Eq \eqref{MA}\\
$\Theta$ &  Eq. \eqref{Thetat} 
\end{tabular}

\bibliography{NonDemolition}
\bibliographystyle{unsrt}

\end{document}